\documentclass[letterpaper,11pt]{article}

\usepackage[english]{babel}
\usepackage[latin1]{inputenc}
\usepackage{amsmath,amssymb,amsthm}
\usepackage{mathrsfs}
\usepackage[letterpaper,top=1.0in,right=1.0in,bottom=1.0in,left=1.0in,centering,dvips]{geometry}
\usepackage[dvips,final]{graphicx}
\usepackage[numbers,sort]{natbib}
\usepackage{authblk}

\newcommand{\longv}[1]{}

\newcommand{\G}{{\mathcal{G}}}
\newcommand{\C}{{\mathcal{C}}}
\newcommand{\R}{{\mathcal{R}}}
\newcommand{\GG}{{G}}

\newcommand{\Pol}{\mbox{\rm Pol}}

\newcommand{\core}[1]{\mathscr{C}(#1)}
\newcommand{\bargset}[1]{\mathscr{B}(#1)}
\newcommand{\kernel}[1]{\mathscr{K}(#1)}

\newcommand{\qtugameKernel}{\mathit{\bs{K}}}
\newcommand{\K}{{_{\bf K}}}
\newcommand{\qtugameBargSet}{\mathit{\bs{BS}}}
\newcommand{\BS}{{_{\bf BS}}}

\newcommand{\bs}[1]{\boldsymbol{#1}}

\newcommand{\tuple}[1]{\langle #1 \rangle}
\newcommand{\tugame}{\tuple{N,\linebreak[0]v}}

\renewcommand{\Re}{\mathbb{R}}
\renewcommand{\emptyset}{\varnothing}
\newcommand{\game}{\mathcal{G}}

\newcommand{\nbdash}{\nobreakdash-\hspace{0pt}}

\newcommand{\PP}{\mbox{\textbf{P}}}
\newcommand{\NP}{\mbox{\textbf{NP}}}

\newcommand{\CONP}{\mbox{\textrm{co\nobreakdash-}\NP}}
\newcommand{\CONPh}{\CONP\nobreakdash-hard}
\newcommand{\CONPc}{\CONP\nobreakdash-com\-plete}

\newcommand{\SigmaP}[1]{\mbox{$\boldsymbol{{\Sigma}_{#1}^{P}}$}}
\newcommand{\PiP}[1]{\mbox{$\boldsymbol{{\Pi}_{#1}^{P}}$}}

\newcommand{\PiPh}[1]{\PiP#1\nobreakdash-hard}

\newcommand{\PiPc}[1]{\PiP#1\nobreakdash-com\-plete}
\newcommand{\DeltaP}[1]{\mbox{$\boldsymbol{{\Delta}_{#1}^{P}}$}}
\newcommand{\DeltaPh}[1]{\DeltaP#1\nobreakdash-hard}
\newcommand{\DeltaPc}[1]{\DeltaP#1\nobreakdash-com\-plete}

\newcommand{\FP}{\mbox{\textbf{FP}}}

\newcommand{\NPSV}{\mbox{\textbf{NPSV}}}
\newcommand{\NPMV}{\mbox{\textbf{NPMV}}}

\newcommand{\CoreMembership}{{\textsc{Core}{-}\textsc{Check}}}
\newcommand{\CoreNonEmptyness}{{\textsc{Core}{-}\textsc{NonEmptiness}}}

\newcommand{\KernelMembership}{{\textsc{Kernel}{-}\textsc{Check}}}
\newcommand{\BargainingSetMembership}{{\textsc{BargainingSet}{-}\textsc{Check}}}
\newcommand{\graph}{\mathcal{GG}}
\newcommand{\marginal}{\mathcal{MCN}}

\theoremstyle{plain}
\newtheorem{theo}{Theorem}[section]
\newtheorem{lemma}[theo]{Lemma}
\newtheorem{prop}[theo]{Proposition}
\newtheorem{corol}[theo]{Corollary}
\newtheorem{fact}[theo]{Fact}

\theoremstyle{definition}
\newtheorem{defin}[theo]{Definition}
\newtheorem{example}[theo]{Example}

\title{On the Complexity of Core, Kernel, and Bargaining Set}

\author{Gianluigi Greco}
\affil{Dipartimento di Matematica,\\
Universit\`a della Calabria, I-87036 Rende(CS), Italy\\
\texttt{ggreco@mat.unical.it}}
\author{Enrico Malizia}
\author{Luigi Palopoli}
\author{Francesco Scarcello}
\affil{DEIS,\\
Universit\`a della Calabria, I-87036 Rende(CS), Italy\\
\texttt{\textbraceleft emalizia,palopoli,scarcello\textbraceright @deis.unical.it}}
\date{}

\begin{document}
\maketitle

\begin{abstract} Coalitional games are mathematical models suited to analyze scenarios where
players can collaborate by forming coalitions in order to obtain higher worths than by acting in isolation. A
fundamental problem for coalitional games is to single out the most desirable outcomes in terms of appropriate
notions of worth distributions, which are usually called {solution concepts}.
Motivated by the fact that decisions taken by realistic players cannot involve unbounded resources,
recent computer science literature reconsidered the definition of such concepts by advocating the relevance of
assessing the amount of resources needed for their computation in terms of their computational complexity.
By following this avenue of research, the paper provides a complete picture of the complexity issues arising with three prominent solution concepts for coalitional games with transferable utility,
namely, the core, the kernel, and the bargaining set,
whenever the game worth-function is represented in some reasonable compact form (otherwise, if the worths of all coalitions are explicitly listed, the input sizes are so large that complexity problems are---artificially---trivial).
The starting investigation point is the setting of graph games, about which various open
questions were stated in the literature. The paper gives an answer to these questions, and in addition provides
new insights on the setting, by characterizing the computational complexity of the three concepts in some
relevant generalizations and specializations.
\end{abstract}

\section{Introduction}\label{sec:intro}

Coalitional games have been introduced by \citet{vonNeumann_Morgenstern:GameTheory} to characterize stable
payoff distributions in scenarios where players can collaborate by forming coalitions in order to obtain higher
worths than by acting in isolation.
A \emph{coalitional game} (with transferable utility) can abstractly be modelled as a pair $\game=\tuple{N,v}$,
where $N$ is a finite set of players, and where $v$ is a function associating with each coalition $S\subseteq N$ a
certain worth $v(S)\in \Re$ that players in $S$ may obtain by collaborating with each other. The outcome of
$\game$ is an \emph{imputation}, i.e., a vector of payoffs ${(x_i)}_{i\in N}$ meant to specify the
distribution of the total worth $v(N)$ granted to each player in $N$. Imputations are required to be
\emph{efficient}, i.e., $\sum_{i\in N}x_i=v(N)$, and \emph{individually rational}, i.e., $x_i\geq v(\{i\})$, for
each $i\in N$. In the following, the set of all imputations of $\game$ is denoted by $X(\game)$.

It is easily seen that, for any given coalitional game $\game$, the set $X(\game)$ might even contain infinitely many
payoff vectors. In this context, a fundamental problem is to single out the most desirable
ones in terms of appropriate notions of worth distributions, which are usually called \emph{solution concepts}.
Traditionally, this question has been studied in economics and game theory in the light of providing arguments
and counterarguments about why such proposals are reasonable mathematical renderings of the intuitive concepts
of fairness and stability. For instance, well-known and widely-accepted solution concepts are the \emph{Shapely
value}, the \emph{core}, the \emph{kernel}, the \emph{bargaining set}, and the \emph{nucleolus} (see, e.g.,
\cite{Osborne_Rubinstein:GameTheory}). Each solution concept defines a set of outcomes that are referred to with
the name of the underlying concept---e.g., the ``core of a game'' is the set of those outcomes satisfying the
conditions associated with the concept of core.

More recently, \citet{Deng_Papadimitriou:ComplexityCooperativeGameSolutionConcepts}
reconsidered the definition of solution concepts for coalitional games from a computer science perspective, by
arguing that decisions taken by realistic agents cannot involve unbounded resources to support reasoning~\cite{Deng_Papadimitriou:ComplexityCooperativeGameSolutionConcepts,Simon:BoundedRationality}, and by suggesting
to formally capture the bounded rationality principle by assessing the amount of resources needed to compute
solution concepts in terms of their computational
complexity~\cite{Deng_Papadimitriou:ComplexityCooperativeGameSolutionConcepts,Kalai_Stanford:FiniteRationality}.
In particular, it has been pointed out that such questions related to the complexity of solution concepts are of particular interest whenever worth
functions are encoded in some succinct way, e.g., they are given in terms of polynomially computable functions~\cite{Deng_Papadimitriou:ComplexityCooperativeGameSolutionConcepts}.
Indeed, under the na\"{i}ve perspective of explicitly
listing all associations of coalitions with their worths, the exponential blow-up of the input representation (w.r.t.\ the number of involved players)
would hide the intrinsic complexity of the solutions concepts.
Coalitional games whose worth functions are encoded without necessarily listing all coalition worths are called \emph{compact games}. In fact, inspired by the approach
of \citet{Deng_Papadimitriou:ComplexityCooperativeGameSolutionConcepts}, the computational complexity of various solution concepts over different classes of compact games has intensively been studied in the last few years (see, e.g.,
\cite{Ieong_Shoham:MarginalContributionNets,Conitzer_Sandholm:CoreComplexity,Elkind_Goldberg_Wooldridge:Complexity_WTG,Prasad:VotingGames,Bilbao:CoopGamesCombStructure}).

In this paper, we continue along the line of research initiated by Deng and Papadimitriou, by analyzing the
computational complexity of three prominent solutions concepts (the core, the kernel, and the bargaining set)
over classes of compact games. Before illustrating our contribution, we next provide details on the classes of
games considered in the paper and overview the relevant related literature.

\subsection{Compactly Specified Games: Graph Games and Marginal Contribution Nets}

\medskip

\noindent \textbf{Graph Games.} The setting of \emph{graph games} is precisely the one analyzed by \citet{Deng_Papadimitriou:ComplexityCooperativeGameSolutionConcepts}. In graph games, worths for coalitions over a set $N$ of players are
defined based on a weighted undirected {graph} $\GG=\tuple{(N,E),w}$, whose nodes in $N$ correspond to the players, and where the list $w$
encodes the edge weighting function, so that $w(e)\in \mathbb{R}$ is the weight associated with the edge $e\in E$. Then, the worth of an
arbitrary coalition $S\subseteq N$ is defined as the sum of the weights associated with the edges contained in $S$, i.e., as the value
$v(S)=\sum_{e\in E \mid e\subseteq S} w(e)$.

\begin{example}\label{ex:graphGame}
Consider the graph game in Figure~\ref{fig:EsempioGraphGame} formed by the players in $\{a,b,c,d\}$.
\begin{figure}[t]
\centering
\includegraphics[width=0.27\textwidth]{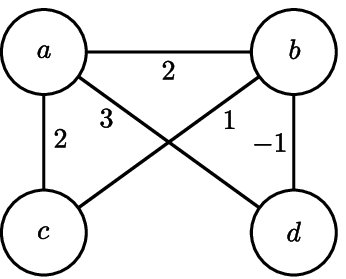}\vspace{-1mm}
\caption{The graph game in Example~\ref{ex:graphGame}.}\label{fig:EsempioGraphGame}\vspace{-3mm}
\end{figure}
There, the coalition $\{a,b\}$ gets a worth $v(\{a,b\})=2$, while the coalition $\{a,b,d\}$ gets a worth
$v(\{a,b,d\})=3+2-1=4$. Note that the graph encodes $2^4$ coalition worths, via 5 weights only. In
such a representation, $O(n^2)$ weights succinctly encode the $2^n$ coalition worths, where $n$ is the number of players.~\hfill~$\lhd$
\end{example}

Within the setting of graph games, \citet{Deng_Papadimitriou:ComplexityCooperativeGameSolutionConcepts} characterized the intrinsic complexity of
various tasks, mainly focusing on problems related to the core. For instance, they showed that checking whether
the core is non-empty
and that checking whether a payoff vector belongs to the core are $\CONP$-complete problems. Moreover, they
provided a polynomial-time computable closed-form characterization for the Shapely value, and showed that this
value coincides with the (pre)nucleolus.
And, finally, they completed the picture of the complexity issues arising with graph games by showing the
$\NP$-hardness of deciding whether a payoff vector belongs to the bargaining set and by \emph{conjecturing} the
following two results on graph games:

\begin{itemize}
  \item[\textbf{(C1)}] {\em Deciding whether a payoff vector belongs to the bargaining set is $\PiP{2}$-complete; and}

  \item[\textbf{(C2)}] {\em Deciding whether a payoff vector belongs to the kernel is $\NP$-hard.}
\end{itemize}

\medskip

\noindent \textbf{Marginal Contribution Nets.} A class of compact games that received considerable attention in
the last few years is that of (games encoded via) \emph{marginal contribution networks}, proposed by \citet{Ieong_Shoham:MarginalContributionNets}.
A marginal contribution network (short: MC\nbdash net) is constituted by a set of rules, each one having the
form: $\{pattern\}\to value$, where a $pattern$ is a conjunction that may include both positive and negative
literals, with each literal denoting a player, and $value$ is the additive contribution associated with this pattern. A rule is said to {apply} to a coalition $S$ if all the players whose literals occur positively in the
pattern belong to $S$, and all the players whose literals occur negatively in the pattern do not belong to $S$. When more than one rule applies to a coalition, the value for that coalition is given by the contribution
of all those rules, i.e., by the sum of their values. If no rule applies, then the value for the coalition is
set to zero, by default.

\begin{example}
Consider the following marginal contribution network inclunding the following three rules $\{a\land b\}\to 5$,
$\{b\}\to 2$, and  $\{a\land\lnot b\}\to3$, over the players in $\{a,b\}$.

The network defines a coalitional game such that:  $v(\{a\})=3$ (the third rule applies), $v(\{b\})=2$ (the
second rule applies), and $v(\{a,b\})=5+2=7$ (both the first and the second rules apply, but not the third one).
~\hfill~$\lhd$
\end{example}

The ``structure'' of player interactions in MC\nbdash nets is represented
 via their associated \emph{agent graphs}~\cite{Ieong_Shoham:MarginalContributionNets}. The agent graph
associated with an MC\nbdash net is the undirected graph whose nodes are the players of the game, and where for
each rule $\{pattern\}\to value$, the subgraph induced over the players occurring together in $pattern$ is a clique, weighted by $value$.

For any graph game $\GG=\tuple{(N,E),w}$, there is an equivalent MC\nbdash net representation having the same
size, and whose ``structure'' is preserved in that its associated agent graph coincides with $\GG$. Indeed,
given any game specified via the graph $\GG=\tuple{(N,E),w}$, we can just create the rule $\{i\wedge j\}\to
w(e)$, for each edge $e=\{i,j\}\in E$ (cf. \cite{Ieong_Shoham:MarginalContributionNets}).
However, the converse is not true, since MC\nbdash nets allow to express any arbitrary coalitional game, while
graph games are not fully expressive. For instance, graph games cannot model a scenario where a group of agents
$S$ has value of 1 if, and only if, $|S|> |N|/2$ (see \cite{Ieong_Shoham:MarginalContributionNets}, for details
on the expressiveness of the frameworks).

In the light of the above observations, hardness results for graph games immediately hold over marginal
contribution networks. For instance, checking whether a payoff vector is in the core and checking whether the
core is non-empty are $\CONP$-hard problems on MC\nbdash nets. Symmetrically, membership results for MC\nbdash
nets also hold for graph games, even when they are established over some ``structurally'' restricted classes of games
(because the structure is preserved). For instance, on marginal contribution networks associated
with \emph{acyclic} agent graphs or, more generally, with agent graphs having \emph{bounded treewidth}~\cite{RS84}, \citet{Ieong_Shoham:MarginalContributionNets} showed that deciding whether a payoff vector is in the core and deciding the non-emptiness of the core
are feasible in polynomial time; therefore, these feasibility results immediately apply to graph games.

Moreover, \citet{Ieong_Shoham:MarginalContributionNets} proved that checking whether a payoff vector is in the core of a game encoded via
marginal contribution networks is in $\CONP$ (as for the class of graph games). However, they left as an
intriguing open problem the following (which, in fact, does not follow from the corresponding result on graph
games):

\begin{itemize}
  \item[\textbf{(O1)}] {\em Is the problem of deciding core non-emptiness over MC\nbdash nets feasible in $\CONP$?}
\end{itemize}

Indeed, they observed that the ``obvious'' certificate of non-emptiness of the core based on the
Bondareva-Shapley theorem is exponential in size. Thus, it cannot help deriving the corresponding membership in
$\CONP$. And, in fact, different technical machineries have to be used for facing the question.

\subsection{Contributions}

In this paper, we analyze the computational complexity of the core, the kernel, and the bargaining set over
graph games and marginal contribution networks. In particular, we show that conjectures \textbf{(C1)} and
\textbf{(C2)} by \citeauthor{Deng_Papadimitriou:ComplexityCooperativeGameSolutionConcepts} are correct, and we provide a positive answer to question
\textbf{(O1)} by \citeauthor{Ieong_Shoham:MarginalContributionNets}. In detail , as our main technical contributions, we show that:

\begin{itemize}
\item[$\blacktriangleright$] On graph games, checking whether a payoff vector is in the kernel is $\NP$-hard,
and actually $\DeltaP{2}$-complete;

\item[$\blacktriangleright$] On graph games, checking whether a payoff vector is in the bargaining set  is
$\PiP{2}$-complete; and,

\item[$\blacktriangleright$] On marginal contribution networks, checking whether the core is non-empty is
feasible in $\CONP$.
\end{itemize}

These main achievements come, however, not alone in the research reported in the paper. Indeed, we go beyond and
study the computational issues arising in relevant \emph{generalizations} and \emph{specializations} of the
setting of graph games.

\begin{description}
\item[Generalizations:] We consider an abstraction of compact coalitional games that is based on assuming that
the worth function is provided as an \emph{oracle} operating over some given structure encoding the game.
In particular, we analyze the cases where each oracle call requires deterministic and non-deterministic
polynomial time (w.r.t.\ the size of the game), respectively; that is, the cases where the oracle encodes a function in $\FP$ or in $\NPSV$, respectively. As an example, graph games and marginal contribution
networks are just two instances of the former setting. Instead, $\NPSV$ oracles are much more powerful, since
they can be used to encode games where coalition worths are the result of complex algorithmic procedures, such
as scheduling, planning or routing activities.

Within this abstract setting, we show that nothing has to be paid for the succinctness of the specifications,
since all the membership results that hold for graph games also hold for any such class as of compact games.
Notably and surprisingly, this is true not only for games whose worth functions can be computed by a deterministic
polynomial-time oracle, but even for the cases where $\NPSV$ oracles are used. Indeed, we show that:

\begin{itemize} \item[$\blacktriangleright$] On games defined via non-deterministic polynomial time
oracles, the problems of checking whether a payoff vector belongs to the core, the kernel, and the bargaining
set are (still) feasible in $\CONP$, $\DeltaP{2}$, $\PiP{2}$, respectively; and

\item[$\blacktriangleright$] On games defined via non-deterministic polynomial time oracles, checking whether
the core is non-empty is (still) feasible in $\CONP$.
\end{itemize}

\item[Specializations:] Finally, by following the perspective adopted by \citet{Ieong_Shoham:MarginalContributionNets},
we analyze the complexity of the kernel in ``structurally restricted'' graph games. In particular, we show that
this concept can be expressed in terms of an optimization problem over \emph{Monadic Second Order Logic (MSO)}
formulae. Based on this encoding and by exploiting Courcelle's Theorem \cite{Courcelle90} and its generalization
to optimization problems due to~\citet*{ALS91}, we show that:

\begin{itemize}
  \item[$\blacktriangleright$] On graph games having bounded treewidth, checking whether a payoff vector is in the kernel is feasible
      in polynomial time
      (w.r.t. the size of the game measured as the number of its nodes, plus the number of its edges, plus all the values of the
      weights associated with it---or, equivalently, plus the number of the bits necessary to store such weights in unary notation).
\end{itemize}
\end{description}

\begin{figure}[t]\centering {\small

 \begin{tabular}{|l||c|c|c|}
   \hline
   {\bf Problem}          & {\bf Graph Games} & {\bf MC\nbdash nets} & {\bf General}\\ \hline\hline

   \CoreMembership          & $\CONP$-complete$^\ddagger$ \cite{Deng_Papadimitriou:ComplexityCooperativeGameSolutionConcepts} & $\CONP$-complete$^\ddagger$ \cite{Ieong_Shoham:MarginalContributionNets} &$\CONP$-complete \\
   \CoreNonEmptyness        & $\CONP$-complete$^\ddagger$ \cite{Deng_Papadimitriou:ComplexityCooperativeGameSolutionConcepts} & $\CONP$-complete$^{\ddagger,\star}$ & $\CONP$-complete\\
   \KernelMembership        & $\DeltaP{2}$-complete$^\ddagger$ & $\DeltaP{2}$-complete & $\DeltaP{2}$-complete\\
   \BargainingSetMembership & $\PiP{2}$-complete & $\PiP{2}$-complete & $\PiP{2}$-complete \\
   \hline
 \end{tabular}\vspace{-1mm}

} \caption{Summary of results. Hardness results on checking problems hold even if the payoff vector is actually
an imputation. $^\ddagger$Feasible in polynomial time on structures having bounded treewidth. $^\star$Hardness
and feasibility in polynomial time on structures having bounded treewidth are shown in
\cite{Ieong_Shoham:MarginalContributionNets}.}\label{fig:Summary}\vspace{-4mm}
\end{figure}

\medskip

A summary of our results is reported in Figure~\ref{fig:Summary}, where \CoreMembership, \KernelMembership, and
\BargainingSetMembership\ denote the problem of checking whether a payoff vector belongs to the core, the
kernel, and the bargaining set, respectively, and where \CoreNonEmptyness\ denotes the problem of checking core
non-emptiness---in fact, recall that kernel and bargaining set are always not empty
 (unless there is no imputation at all)~\cite{{Osborne_Rubinstein:GameTheory}}, and hence the corresponding non-emptiness problem is immaterial for
them.

\subsection{Organization} The rest of the paper is organized as follows. Preliminaries on computational
complexity are reported in Section~\ref{sec:prelim_comp_compl}. An abstract framework for compact games is
discussed in detail in Section~\ref{sec:compact}. The complexity of core, kernel, and bargaining set is analyzed
in Sections~\ref{sec:core}, \ref{sec:kernel}, and \ref{sec:bargainingSet}, respectively. Eventually, a few final
remarks and discussions on some open problems are reported in Section~\ref{sec:conclusion}.

\section{Preliminaries on Computational Complexity}\label{sec:prelim_comp_compl}

In this section we recall some basic definitions about complexity theory, by referring the reader to the work of
\citet{Johnson:CatalogComplexityClasses} for more on this.

\subsection{The Complexity of Decision Problems: {\PP}, {\NP}, and {\CONP}}

\emph{Decision} problems are maps from strings (encoding the input instance over a fixed alphabet, e.g., the
binary alphabet $\{0,1\}$) to the set $\{ ``yes" , ``no" \}$. The class $\PP$ is the set of decision problems
that can be solved by a deterministic Turing machine in polynomial time with respect to the input size, that is,
with respect to the length of the string that encodes the input instance. For a given input $x$, its size is
usually denoted by $||x||$.

Throughout the paper, we shall often refer to computations done by \emph{non-deterministic} Turing machines,
too. Recall that these are Turing machines that, at some points of the computation, may not have one single next
action to perform, but a \emph{choice} between several possible next actions.
A non-deterministic Turing machine answers a decision problem if, on any input $x$, (\emph{i}) there is at least one sequence of choices leading to halt in an accepting state
 if $x$ is a ``yes'' instance (such a sequence is called accepting computation path); and
(\emph{ii}) all possible sequences of choices lead to a rejecting state if $x$ is a ``no'' instance.

The class of decision problems that can be solved by non-deterministic Turing machines in polynomial time is
denoted by $\NP$.
Problems in $\NP$ enjoy a remarkable property: any ``yes'' instance $x$ has a \emph{certificate} of it being a
``yes'' instance, which has polynomial length and which can be checked in polynomial time (in the size $||x||$).
As an example, deciding whether a Boolean formula $\Phi$ over the variables $X_1,\dots,X_n$ is satisfiable,
i.e., deciding whether there exists some truth assignment to the variables making $\Phi$ true, is a well-known
problem in $\NP$; in fact, any satisfying truth assignment for $\Phi$ is obviously a certificate that $\Phi$ is
a ``yes'' instance, i.e., that $\Phi$ is satisfiable.

The class of problems whose complementary problems are in $\NP$ is denoted by $\CONP$. As an example, the
problem of deciding whether a Boolean formula $\Phi$ is \emph{not} satisfiable is in $\CONP$. Of course, the
class $\PP$ is contained in both $\NP$ and $\CONP$.

\subsection{Further Complexity Classes: The Polynomial Hierarchy}

Throughout the paper, we shall also refer to a type of computation called computation with \emph{oracles}.
Intuitively, oracles are subroutines which are supposed to have unary cost.

The classes $\SigmaP{k}$, $\PiP{k}$, and $\DeltaP{k}$, forming the \emph{polynomial hierarchy}, are defined as
follows: $\SigmaP{0} = \PiP{0} = \PP$ and for all $k\ge 1$, $\SigmaP{k}=\NP^{\Sigma^P_{k-1}}$,
$\DeltaP{k}=\PP^{\Sigma^P_{k-1}}$, and $\PiP{k}=\text{co-}\SigmaP{k}$ where $\text{co-}\SigmaP{k}$ denotes the
class of problems whose complementary problem is solvable in $\SigmaP{k}$. Here, $\SigmaP{k}$ (resp.,
$\DeltaP{k}$) models computability by a non-deterministic (resp., deterministic) polynomial-time Turing machine
that may use an oracle in $\SigmaP{k-1}$. Note that $\SigmaP{1}$ coincides with $\NP$, and that $\PiP{1}$
coincides with $\CONP$.

A well-known problem at the $k$-th level of the polynomial hierarchy is deciding the validity of a quantified
Boolean formula with $k$ quantifier alternations. A quantified Boolean formula (short: QBF) with $k$ quantifier
alternations has the form $Q_1\bar X_1Q_2\bar X_2...Q_k\bar X_k \Phi$, where $k\geq 1$, $\bar X_i$ ($1\leq i\leq
k$) is a set of variables, $Q_i\in \{\exists,\forall\}$ ($1\leq i\leq k$), $Q_i\neq Q_{i+1}$ ($1\leq i< k$), and
$\Phi$ is a Boolean formula over the variables in $\bigcup_{i=1}^k \bar X_i$.
The set of all quantified Boolean formulas with $k$ quantifier alternations and $Q_1=\exists$ (resp.,
$Q_1=\forall$) is denoted by QBF$_{k,\exists}$ (resp., QBF$_{k,\forall}$). Deciding the validity of a quantified
Boolean formula in QBF$_{k,\exists}$ (resp., QBF$_{k,\forall}$) is a well-known problem in  $\SigmaP{k}$ (resp.,
$\PiP{k}$). Note that for $k=1$, this problem coincides with the problem of deciding whether the Boolean formula
$\Phi$ is satisfiable (resp., not satisfiable), which is indeed in $\NP$ (resp., $\CONP$).

\subsection{Reductions among Decision Problems}\label{sec:reduction}

A decision problem $A_1$ is \emph{polynomially reducible} to a decision problem $A_2$, denoted by $A_1 \leq_p
A_2$, if there is a polynomial time computable function $h$ (called reduction) such that, for every $x$, $h(x)$ is defined and $x$
is a ``yes'' instance of $A_1$ if, and only if, $h(x)$ is a ``yes'' instance of $A_2$. A decision problem $A$ is
\emph{complete} for a class $\mathcal{C}$ of the polynomial hierarchy (at any level $k\geq 1$, i.e., beyond $\PP$) if $A$ belongs to
$\mathcal{C}$ and every problem in $\mathcal{C}$ is polynomially reducible to $A$. Thus, problems that are
complete for $\mathcal{C}$ are the most difficult problems in $\mathcal{C}$.

It is worthwhile observing that all problems mentioned in this section are known to be complete for the classes
in which their membership has been pointed out. In particular, deciding the validity of a QBF$_{k,\exists}$
(resp., QBF$_{k,\forall}$) formula is the prototypical $\SigmaP{k}$-complete (resp., $\PiP{k}$-complete)
problem.

\subsection{Complexity Classes of Functions}

All the aforementioned problems are decision ones, but very often we are interested in \emph{search} problems where, for any given instance, a (non-Boolean) solution must be computed. The complexity classes of functions allow us to distinguish such problems according to their intrinsic difficulties, which is particularly relevant when their associated decision problems belong to the same complexity class.

Let a finite {\em alphabet} $\Sigma$ with at least two elements be
given. A {\em (partial) multivalued} ({\em MV}) {\em function} $f
:\Sigma^* \mapsto \Sigma^*$ associates  no, one or several
outcomes ({\em results}) with each input string. Let $f(x)$ stand for the set of possible
results of $f$ on an input string $x$; thus, we write $y \in
f(x)$ if $y$ is a value of $f$ on the input string $x$.
Define {\em dom}$(f) = \{ x \mid \exists y (y \in f(x)) \}$ and {\em
graph}$(f) = \{ \langle x,y \rangle \mid x\in dom(f),~y \in f(x)
\}$. If $x \not\in${\em dom}$(f)$, we will say that $f$ is
undefined at $x$. The function $f$ is {\em total} if $dom(f) = \Sigma^*$.

An MV function $f$ is {\em polynomially balanced} if, for each $x$,
the size of each result in $f(x)$ is polynomially bounded in the
size of $x$.

The class $\NPMV$ is defined as the set of all MV functions $f$ such that both
(i) $f$ is polynomially balanced and
(ii) {\em graph}$(f)$ is in \NP.
By analogy, the class $\NPMV_g$ is defined as the class of all polynomially-balanced
multivalued functions $f$ for which {\em graph}$(f)$ is
in $\PP$.
If we deal with {\em (partial) single\nbdash valued functions}, we get the corresponding classes  $\NPSV$ and $\NPSV_g$, respectively~\cite{selm-94}.

A {\em transducer} is a (possibly, non-deterministic) Turing
machine $T$ on the alphabet $\Sigma$ with a read-only input tape, a read-write work tape,
and a write-only output tape.
For any string $x \in \Sigma^*$, we say that $T$ accepts $x$
if $T$ has an accepting computation-path on $x$.
For each $x\in \Sigma^*$ accepted by $T$, we
denote by $T(x)$ the set of all strings that are written by $T$ on the output tape
in its accepting computation-paths on input string $x$.
Thus, every transducer is associated with some MV function $f$
(we say that $T$ {\em computes} $f$) such that, for each $x\in \Sigma^*$,
$f(x)=T(x)$ if $x$ is accepted by $T$; otherwise, $f$ is undefined at $x$ (i.e., $x \not \in dom(f)$).

Note that the class $\FP$ consists of all functions that are
computed by deterministic Turing transducers in polynomial time.

Functions in $\NPMV$ are
characterized in terms of Turing machines as follows.
\begin{fact}
An MV function is in $\NPMV$ if and only if it is computed by a
nondeterministic transducer in polynomial time.
\end{fact}

It is worthwhile noting the difference between the two classes $\NPMV_g$ and $\NPMV$, which contains more complex functions than $\NPMV_g$ (assuming $\PP\neq \NP$). For instance, consider the problem of computing the partial MV function $f_{H}$ that, given a graph $G$, outputs the Hamiltonian cycles of $G$ (if any).
 This function is in $\NPMV_g$ since $graph(f_{H})$ is polynomially balanced and decidable in deterministic polynomial time (for any pair $\tuple{G,C}$, just check whether $C$ is a Hamiltonian cycle of $G$).
Let us consider now the weighted version of this problem, where the input graph $G$ is edge-weighted, and the function values are the weights of Hamiltonian cycles of $G$ (if any).
 Then, this partial MV function, say $f_{W\!H}$, belongs to $\NPMV$ but not to $\NPMV_g$
   (unless $\PP = \NP$).
   Indeed, deciding whether a given pair $\tuple{G,w}$ (graph,weight) belongs to $graph(f_{W\!H})$ is clearly $\NP$-complete (one needs some Hamiltonian cycle having weight $w$ to recognize it as a correct function value).

\section{A Formal Framework for Compact Representations}\label{sec:compact}

Graph games and marginal contribution networks are two prominent examples of coalitional games whose worth functions are defined in terms of some suitable (combinatorial) structure, instead of
 listing the worths of all coalitions.
  In this section, we generalize these two schemes by formalizing the notion of compact representation for coalitional games.

Any {\em compact representation} $\R$ defines suitable encodings for a set of coalitional games, denoted by $\C(\R)$. If all games may be represented by $\R$, we say that this representation is
 {\em complete}.
 For any game $\G\in \C(\R)$, $\R$ defines an encoding $\xi^\R(\G)$ of the game, and
  a worth function $v^\R(\cdot,\cdot)$ that, given $\xi^\R(\G)$ and a set $S$ of players of  $\G$, outputs the worth  $v^\R(\xi^\R(\G),S)$ associated with the coalition $S$ according to $\G$.

For instance, consider the graph-games representation $\graph$: any game $\G\in \C(\graph)$ is encoded as a weighted graph $\xi^\graph(\G)$, and the worth function $v^\graph(\xi^\graph(\G),S)$ is computed for every coalition $S$ by taking the sum of the weights of all edges of $\xi^\graph(\G)$ included in $S$.

 For the case of the marginal-contribution nets compact-representation $\marginal$, any game $\G$ is encoded by a set of rules $\xi^\marginal(\G)$, and the worth-function $v^\marginal(\xi^\marginal(\G),S)$ computes the worth of $S$ as the sum of the values of those rules fired by players in $S$.

\begin{defin}\label{def:compactGeneralForm}
We say that $\R$ is a \emph{polynomial\nbdash time compact representation} (short: \PP{}\nbdash representation) if the
\emph{worth-function} $v^\R(\cdot,\cdot)$ belongs to $\FP$, i.e., it is polynomial\nbdash time computable by a deterministic transducer.

We say that $\R$ is a \emph{non\nbdash deterministic polynomial\nbdash time compact representation} (short:
\NP{}\nbdash  representation) if the \emph{worth-function} $v^\R(\cdot,\cdot)$ belongs to $\NPSV$, i.e., it is polynomial\nbdash time computable by a non\nbdash deterministic transducer.\hfill $\Box$
\end{defin}

Note that both $\graph$ and $\marginal$ are based on worth functions that are efficiently computable, and in fact both of them are polynomial\nbdash time compact representations.

 An interesting feature of a compact representation is its {\em expressive power}.
  Firstly, one may ask whether $\R$ is complete or not, that is, whether it is the case that every coalitional game belongs to $\C(\R)$. For instance $\marginal$ is a complete representation, while $\graph$ is not. However, as a consequence of completeness, for some games $\marginal$ is not able to provide succinct representations. In fact, it is known that there are games whose $\marginal$ encodings have size exponential in the number of players~\cite{Elkind:ExpressiveMCNET}.
  Therefore, for a pair of compact representations $\R_1$ and $\R_2$, one may wonder about the relationship between $\C(\R_1)$ and $\C(\R_2)$, and about the ability to represent games in a more o less succinct way.
Formally, we say that $\R_2$ is at least as expressive (and succinct) as $\R_1$, denoted by $\R_1\precsim_e\R_2$,
  if there exists a function $f$ in $\FP$ that translates a game $\xi^{\R_1}(\G)$ represented in $\R_1$ into an equivalent game $\xi^{\R_2}(\G)$ represented in $\R_2$, that is, into a game with the same worth function as the former one.
  More precisely, we require that, for $\xi^{\R_2}(\G)=f(\xi^{\R_1}(\G))$,  $v^{\R_1}(\xi^{\R_1}(\G),S)=v^{\R_2}(\xi^{\R_2}(\G),S)$, for each coalition of players $S$ in the game $\G$.

For instance, it can easily be shown that $\graph\precsim_e\marginal$. In the rest of the paper, we shall provide complexity results for $\graph$, $\marginal$, and for arbitrary \PP{} and \NP{} compact representations.

For the sake of presentation, whenever a compact representation $\R$ is understood, we just write
 $\G$ instead of $\xi^{\R}(\G)$, and $v(S)$ instead of $v^{\R}(\xi^{\R}(\G),S)$.

\section{The Complexity of the Core}\label{sec:core}

The concept of the \emph{core} goes back to the work of \citet{Edgeworth:Core} and it has been formalized by
\citet{Gillies:Solutions}. To review its definition, we need to state some preliminary concepts and notations,
which will extensively be used throughout the paper.

For any coalition $S\subseteq N$, let $|S|$ denote the cardinality of $S$, and let $\Re^S$ be the $|S|$\nbdash
dimensional real coordinate space, whose coordinates are labeled by the members of $S$; in particular, given a
\emph{payoff vector} $x\in\Re^S$, $x_i$ denotes the component associated with the player $i\in S$.
A vector $x\in \Re^S$ is called an $S$-feasible vector if $\sum_{i\in S} x_i=v(S)$. The value $\sum_{i\in S}
x_i$ will be simply denoted by $x(S)$ in the following.
Let $\game=\tugame$ be a coalitional game, and let $x$ be an imputation taken from the set $X(\game)$ of all
imputations of $\game$---recall from the Introduction that such an $x$ must be {efficient}, i.e., $\sum_{i\in
N}x_i=v(N)$, and {individually rational}, i.e., $x_i\geq v(\{i\})$, for each $i\in N$. The pair $(y,S)$ is an
\emph{objection to $x$} if $y$ is an $S$\nbdash feasible payoff vector such that $y_k> x_k$ for all $k\in S$.

\begin{defin}\label{def:coreTU}
The \emph{core} $\core{\game}$ of a coalitional game $\game=\tuple{N,v}$ is the set of all imputations $x$ to
which there is no objection; that is,
\[
\core{\game}=\left\{x\in X(\game)\left|\nexists S\subseteq N \text{ and } y\in\Re^S  \text{ such that } y(S)=v(S) \text{ and } y_k>x_k,\forall k\in
S\right.\right\}.
\]

\vspace{-4mm}\hfill$\Box$
\end{defin}

Thus, an imputation $x$ in the core is ``stable'' precisely because there is no coalition whose members will
receive a higher payoff than in $x$ by leaving the grand\nbdash coalition.

It is easily seen that Definition~\ref{def:coreTU} can be equivalently restated as the set of all solutions
satisfying the following inequalities \cite[see, e.g.,][]{Osborne_Rubinstein:GameTheory}:
\begin{align}
\sum_{i\in S}x_i&\geq v(S),\quad\forall S\subseteq N \land S\neq\emptyset\label{eq:VincoliCoalizioni}\\
\sum_{i\in N} x_i&\leq v(N)\label{eq:VincoloGrandCoalition}.
\end{align}
\noindent

In particular, the last inequality, combined with its opposite in \eqref{eq:VincoliCoalizioni}, enforces the
efficiency of solutions; moreover, inequalities in \eqref{eq:VincoliCoalizioni} over singleton coalitions
enforce their individual rationality.

\begin{example}\label{ex:CORE_TU}
Let $\game=\tugame$ be a TU game with $N=\{a,b,c\}$, $v(\{a\})=v(\{b\})=v(\{c\})=0$, $v(\{a,b\})=20$,
$v(\{a,c\})=30$, $v(\{b,c\})=40$, and $v(\{a,b,c\})=42$. Consider the imputation $x$ such that: $x_a=4$,
$x_b=14$, and $x_c=24$. Since $v(\{b,c\})=40>38=x(\{b,c\})$ such imputation $x$ is not in $\core{\game}$.

In fact, we can show that $\core{\game}$ is empty. To this end, consider coalitions $S_1=\{a,b\}$, $S_2=\{a,c\}$
and $S_3=\{b,c\}$ and the worths associated with them by the worth function. An imputation $x$ to be in
$\core{\game}$ have to satisfy the following three conditions:
\begin{align*}
x_a+x_b &\geq 20\\
x_a+x_c &\geq 30\\
x_b+x_c &\geq 40.
\end{align*}
Summing up these inequalities we obtain that $2 x_a+2 x_b+2 x_c\geq 90$, implying that $x_a+x_b+x_c\geq 45$. Thus,
the core of $\game$ is empty because the grand\nbdash coalition would need to receive $45$ instead of $v(N)=42$
in order to satisfy the claims of $S_1$, $S_2$ and $S_3$.

Consider, instead, the game $\game'=\tuple{N,v'}$ whose worth function $v'$ is the same as that of $\game$
except for the grand\nbdash coalition for which $v'(N)=45$. Then, it is easily checked that the imputation $x'$
such that $x'_a=5$, $x'_b=15$, and $x'_c=25$ is in $\core{\game'}$.~\hfill~$\lhd$
\end{example}

As it emerged from the above example, the core of a game can be empty. In fact, checking whether this is not the case is a \CONP-hard problem for graph games~\cite{Deng_Papadimitriou:ComplexityCooperativeGameSolutionConcepts}.
It is easy to see that this hardness result can be extended easily to all compact game representations at least as expressive (and succinct) as graph games.

\begin{prop}\label{prop:hardness-core}
Let $\R$ be any compact representation such that $\graph\precsim_e \R$ (e.g., $\R=\marginal$). On the class $\C(\R)$, \CoreNonEmptyness\ is \CONPh.
\end{prop}
\begin{proof}
From the \CONP-hardness for graph games~\cite{Deng_Papadimitriou:ComplexityCooperativeGameSolutionConcepts}, we know that there is a polynomial-time reduction $f_1$ from any \CONP\ problem $\Upsilon$ to the core non-emptiness problem for graph games. Moreover, recall from Section~\ref{sec:compact} that $\graph\precsim_e \R$ means that
there exists a polynomial-time function $f_2$ that translates any graph game $\xi^{\graph}(\G)$ into an equivalent game $f_2(\xi^{\graph}(\G))$ belonging to $\C(\R)$, that is, into a game with the same worth function and thus the same core as the former one. Therefore, the composition of $f_1$ and $f_2$ is a polynomial-time reduction from $\Upsilon$ to the core non-emptiness problem for games in $\C(\R)$.
\end{proof}

In fact, for the case of graph games, the precise complexity of the non-emptiness problem for the core is known, as~\citet{Deng_Papadimitriou:ComplexityCooperativeGameSolutionConcepts} have shown that this problem is \CONPc.
However, it was open whether the membership still holds for
marginal contribution networks~\cite{Ieong_Shoham:MarginalContributionNets} and, possibly, for more general compact representations.
 In this section, we positively answer this question and show that in fact membership in $\CONP$ holds for any class of games $\C(\R)$ associated with an \NP{} representation scheme $\R$.

In order to get this result, we have to identify some succinct certificate
that the core of a game is empty. Indeed,
 it was observed that the ``obvious'' certificate of non-emptiness of the core based on the
Bondareva-Shapley theorem is exponential in size~\cite{Ieong_Shoham:MarginalContributionNets}.
 The technical machinery that we need is originally due to Helly, who proved the
following beautiful result on families of convex sets.

\begin{prop}[Helly's Theorem~\cite{Helly,Rabin:Helly}]\label{theo:hellyTheo}
Let $\mathcal{C}=\{c_1,\dots,c_m\}$ be a finite family of convex sets in $\Re^n$, where $m>n$. If
$\bigcap_{i=1}^m c_i=\emptyset$, there is a set $\mathcal{C'}\subseteq\mathcal{C}$, such that
$|\mathcal{C'}|=n+1$ and $\bigcap_{c_i\in\mathcal{C'}}c_i=\emptyset$.
\end{prop}

However, we next provide a direct proof of the existence of such succinct certificates of infeasibility.
Indeed, we believe that our proof provides a nice geometrical interpretation of
Helly's Theorem for the convex sets related with the core, and thus it may be of independent interest for the reader.
 The approach is based on the well-known fact that, because of inequalities \eqref{eq:VincoliCoalizioni} and
\eqref{eq:VincoloGrandCoalition}, the core of a coalitional game over $N$ players is a polyhedral set of $\Re^N$.

\subsubsection{Preliminaries on Polyhedral Sets}\label{polyhedra:prelim}

We next give some useful definitions and facts about polyhedral sets. We refer the interested reader to any text
on this subject for further readings (see, e.g., \citep{Grunbaum:ConvexPolytopes,Brondsted:ConvexPolytopes}).

Let $n>0$ be any natural number. A \emph{Polyhedral Set} (or \emph{Polyhedron}) $P$ of $\Re^n$ is the intersection of a finite set $\mathcal{S}$ of closed halfspaces of $\Re^n$. Note that in this paper we always assume, unless otherwise stated, that $n>0$. We denote this polyhedron by $\Pol(\mathcal{S})$. 

Recall that a \emph{hyperplane} $H$ of $\Re^n$ is a set of points $\{x\in\Re^n|a^Tx=b\}$, where $a\in\Re^n$ and $b\in\Re$. The closed \emph{halfspace} $H^+$ is the set of points $\{x\in\Re^n|a^Tx \geq b\}$. We say that these points \emph{satisfy} $H^+$. We denote the points that do not satisfy this halfspace by $H^-$, i.e., $H^- = \Re^n \setminus H^+ = \{x\in\Re^n|a^Tx < b\}$. Note that $H^-$ is an open halfspace. We say that $H$ \emph{determines} $H^+$ and $H^-$. Define the \emph{opposite} of $H^+$ as the set of points $\bar{H}^+= \{x\in\Re^n|a'^Tx\geq b'\}$, where $a' = -1 \cdot a$ and $b' = -1\cdot b$. Note that $\bar{H}^+ = H^- \cup H$, since it is the set of points $\{x\in\Re^n|a^Tx \leq b\}$.

Let $P=\Pol(\mathcal{S})$ be a polyhedron and $H$ a hyperplane. Then, $H$ \emph{cuts} $P$ if both $H^+$ and $H^-$ contain points of $P$, and we say that $H$ \emph{passes through} $P$, if there is a non-empty touching set $C = H\cap P$. Furthermore, we say that $H$ \emph{supports} $P$,  or that it is a \emph{supporting hyperplane} for $P$, if $H$ does not cut $P$, but passes through $P$, i.e., it just touches $P$, as the only common points of $H$ and $P$ are those in their intersection $C$.

Moreover, we say that $H^+$ is a \emph{supporting halfspace} for $P$ if $H$ is a supporting hyperplane for $P$ and $P\subseteq H^+$. Note that $P \subseteq \Pol(\mathcal{S}')$ for any set of halfspaces $\mathcal{S}'\subseteq \mathcal{S}$, since the latter polyhedron is obtained from the intersection of a smaller set of halfspaces than $P$. We say that such a polyhedron is a \emph{supporting polyhedron} for $P$.

Recall that, for any set $A \subseteq \Re^n$, its dimension $\dim(A)$ is the dimension of its affine hull. For
instance, if $A$ consists of two points, or it is a segment, its affine hull is a line and thus $\dim(A)=1$. By
definition, $\dim(\emptyset)=-1$, while single points have dimension~0.
For any hyperplane $H$, $\dim(H)=n-1$, while
the intersection $C$ of any pair of (non-parallel) hyperplanes $H_1$ and $H_2$ has dimension $n-2$.

Every hyperplane $H$ has precisely two normal vectors, while its associated halfspace $H^+$ has one normal vector, that is, the normal vector of $H$ that belongs to $H^+$.
 The \emph{dihedral angle} $\delta(H_1^+,H_2^+)$ between two non-parallel halfspaces $H_1^+$ and $H_2^+$ is the smallest angle between the corresponding normal vectors. For such halfspaces, in this paper we always consider rotations whose axis is the (affine) subspace $C=H_1\cap H_2$ of dimension $n-2$, so we avoid to explicitly mention rotation axes, hereafter. Note that such rotations are uniquely identified by their amount (angle) of rotation, which is the one degree of freedom.
 Formally, define the result of a \emph{rotation} of a halfspace $H_1^+$ towards a halfspace
  $H_2^+$ of an angle $-\pi<\alpha<\pi$ to be the halfspace $H_3^+$ such that
   $\delta(H_3^+,H_2^+)= \delta(H_1^+,H_2^+)-\alpha$, and $H_i\cap H_j=C$, $\forall i,j\in\{1,2,3\}$, $i\neq j$ (the points on the axis $C$ of rotation are fixed).

A set $F \subseteq P$ is a \emph{face} of $P$ if either $F=\emptyset$, or $F=P$, or if there exists a supporting hyperplane $H_F$ of $P$ such that $F$ is their touching set, i.e., $F=H_F\cap P$. In the latter case, we say that $F$ is a \emph{proper face} of $P$. A \emph{facet} of $P$ is a proper face of $P$ having the largest possible dimension, that is, whose dimension is $\dim(P)-1$.

The following facts are well known \citep{Grunbaum:ConvexPolytopes}:
\begin{enumerate}
\item For any facet $F$ of $P$, there is a halfspace $H^+\in \mathcal{S}$ such that $F=H^+\cap P$. We say that $H^+$ \emph{generates} $F$.\label{fact:generator}
\item For any proper face $F$ of $P$, there is a facet $F'$ of $P$ such that $F\subseteq F'$.\label{fact:facet}
\item If $F$ and $F'$ are two proper faces of $P$ and $F \subset F'$, then $\dim(F) < \dim(F')$.\label{fact:dimension}
\end{enumerate}

\subsubsection{Separating Polyhedra from a Few Supporting Halfspaces} \label{sec:separating}

We start with a pair of technical results.
With a little abuse of notation, since coalitions correspond to the inequalities \eqref{eq:VincoliCoalizioni} and hence with the
associated halfspaces of $\Re^n$, we use hereafter interchangeably these terms.

\begin{figure}[t]
\centering
\includegraphics[width=0.5\textwidth]{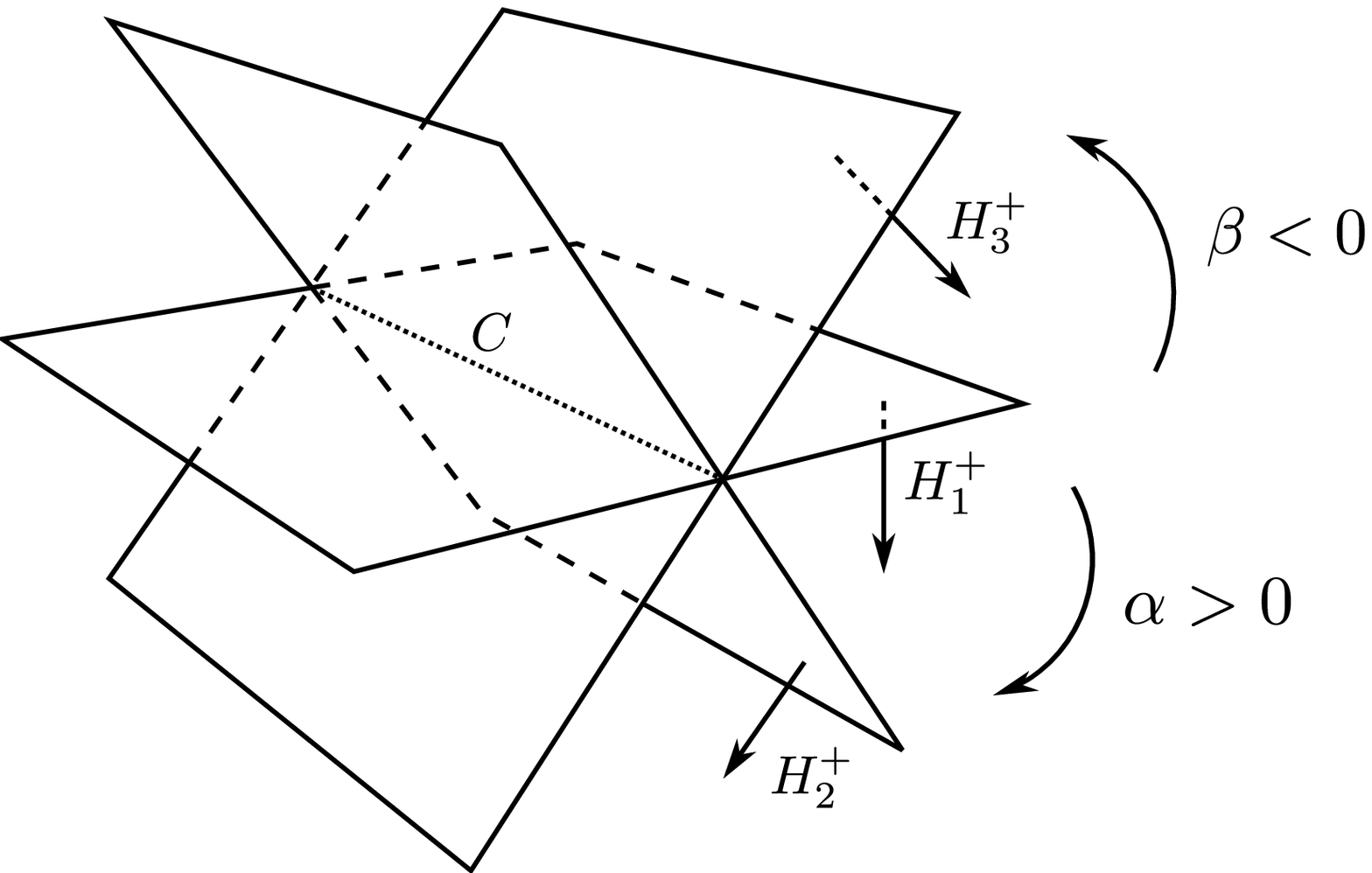}
\caption{Rotations of halfspaces in Lemma~\ref{lemma:RoofLemma}.}\label{fig:roof}
\end{figure}

\begin{lemma}[Roof Lemma]\label{lemma:RoofLemma}
Let $H_1^+$, $H_2^+$, and $H_3^+$ be three halfspaces such that $H_i\cap H_j=C\neq\emptyset$, $\forall i,j\in\{1,2,3\}$, $i\neq j$, and such that $H_3^+$ may be obtained by  rotating $H_1^+$ towards $H_2^+$ by $\beta <0$, with $\delta(H_2^+,H_1^+)-\beta<\pi$. Then, $H_1^+$ is a
supporting halfspace for $H_2^+ \cap H_3^+$, i.e., for $\Pol(\{H_2^+,H_3^+\})$.
\end{lemma}
\begin{proof}
Note that $H_2^+$ may be obtained by  rotating $H_1^+$ towards $H_2^+$ by their dihedral angle $0<\alpha=\delta(H_2^+,H_1^+) <\pi$. Thus, all points $A$ of $H_1\setminus C$ that
belong to $H_2^+$ are in $H_3^-$, because $H_3$ is obtained by rotating those points (in the
opposite direction) by the angle $\beta<0$. Symmetrically, all points $B$ of $H_1\setminus C$
belonging to $H_3^+$ are in $H_2^-$, see Figure~\ref{fig:roof}, for a three\nbdash dimensional illustration.
 Moreover, observe that all points in $H_1\setminus C$ are involved in the rotations and thus belong to either $A$ or $B$.
 It follows that $(H_1\setminus C)\cap (H_2^+\cap H_3^+)=\emptyset$, whence $H_1$ is a supporting hyperplane.
 Finally, since $\delta(H_2^+,H_1^+)-\beta<\pi$, it is easy to see that some point in $H_2^+\cap H_3^+$ is in  $H_1^+$, and thus $H_1^+$ is in fact a supporting halfspace for $\Pol(\{H_2^+,H_3^+\})$.
\end{proof}

We next show that, whenever a full-dimensional polyhedron $P=\Pol(\mathcal{S})$ of $\Re^n$ is separated from some hyperplane $H_P$, there exists a subset of at most $n$ halfspaces corresponding to facets of $P$ that define  a larger polyhedron of $\Re^n$ (a rough approximation of $P$), which is still separated from $H_P$.
We first give the proof idea, with the help of Figure~\ref{fig:certificate}. For the sake of intuition, imagine that $H_P^+$ is the inequality~(\ref{eq:VincoloGrandCoalition}) associated with the grand\nbdash coalition, while the set $\mathcal{S}$ corresponds to the other inequalities~(\ref{eq:VincoliCoalizioni}), whence in this case the core is empty.
If there is a facet $F$ of $P$ that is parallel to $H_P$, we are trivially done, because its associated halfspace already provides the desired separated polyhedron, succinctly described by just one inequality.
Otherwise such a face $F$ has a smaller dimension but, from Fact~\ref{fact:facet}, there exists a facet $F'$ of $P$ such that $F\subset F'$. In the three-dimensional example shown in Figure~\ref{fig:certificate}, $F$ is the vertex at the bottom of the diamond, $H_F^+$ is the halfspace (anti-)parallel to $H_P$ that contains $F$ (but it is not associated with any inequality generating $P$), and $F'$ is some facet on its ``dark side.'' Let $C=H_F\cap H_{F'}$, and consider the rotation of $H_F^+$ towards $H_{F'}^+$ of a negative angle $\beta$ (we go on the opposite direction w.r.t.\ $H_{F'}^+$) that first touches $P$, say $H_{F''}$. As illustrated in Figure~\ref{fig:certificate}, where the face $F''$ is an edge of the diamond, $F''$ properly includes $F$ and its dimension is at least $d> \dim(F)$. From the above lemma, $H_F^+$ is a supporting halfspace of the polyhedron $\Pol(H_{F'}^+,H_{F''}^+)$, called its roof, which contains $P$ and is separated from $H_P$.
 However, we are not satisfied because we would like that such a polyhedron is described by (at most $n$) halfspaces {\em taken from} $\mathcal{S}$, and $H_{F''}^+$ does not belong to $\mathcal{S}$, in general (as in our example, where it does not generate a facet of $P$. Then, we proceed inductively, by observing that $H_{F''}^+$ is a supporting halfspace for $\Pol(H_{F''_1}^+, H_{F''_2}^+)$---its roof, whose faces have higher dimension than $F''$. In the running example, they are both facets of the diamond, and hence the property immediately holds. In general, the procedure may continue, encountering each time at least one facet and one more face with a higher dimension than the current one. The formal proof follows.

\begin{figure}[t]
\centering
\includegraphics[width=0.3\textwidth]{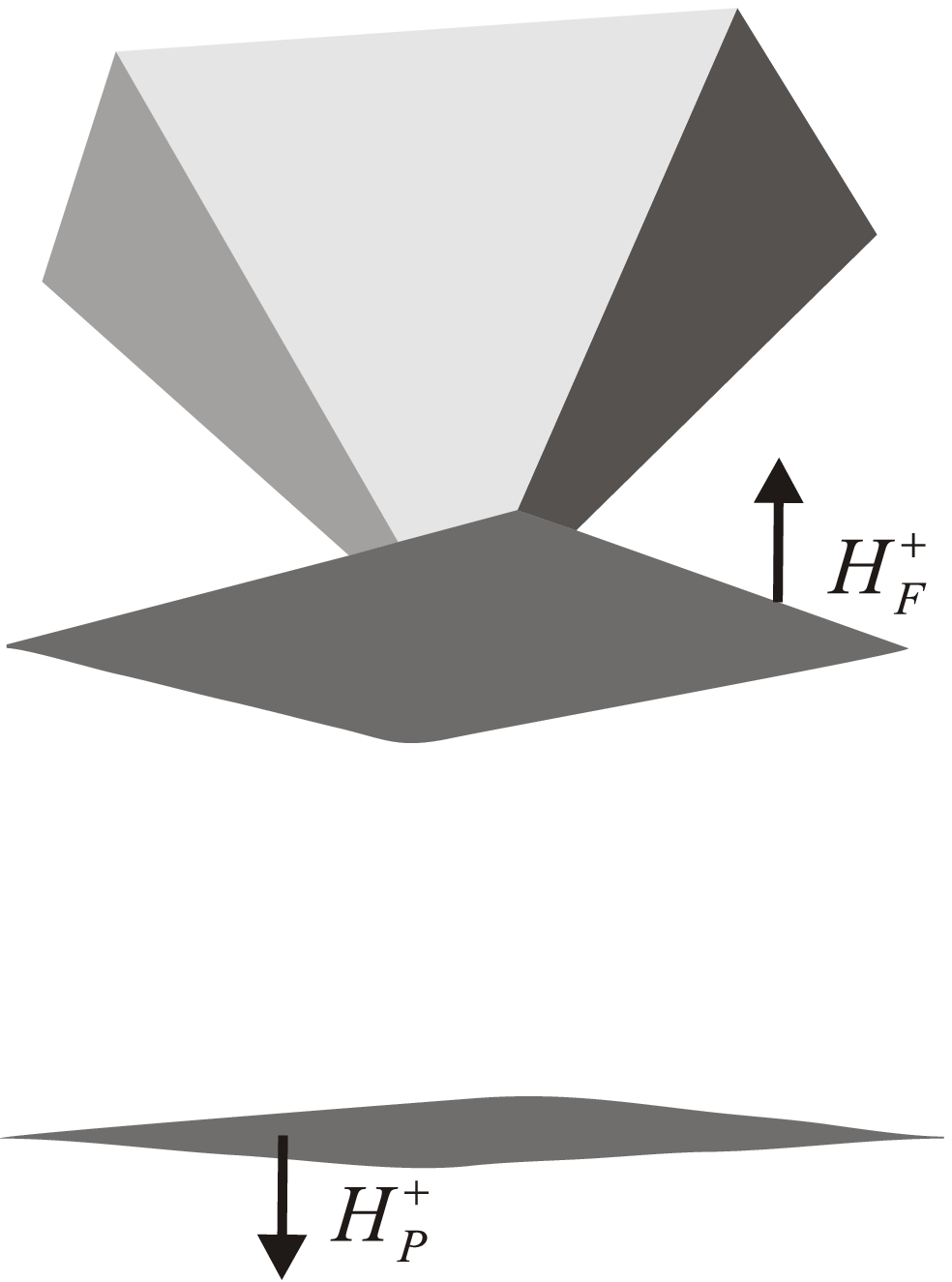}
\includegraphics[width=0.342\textwidth]{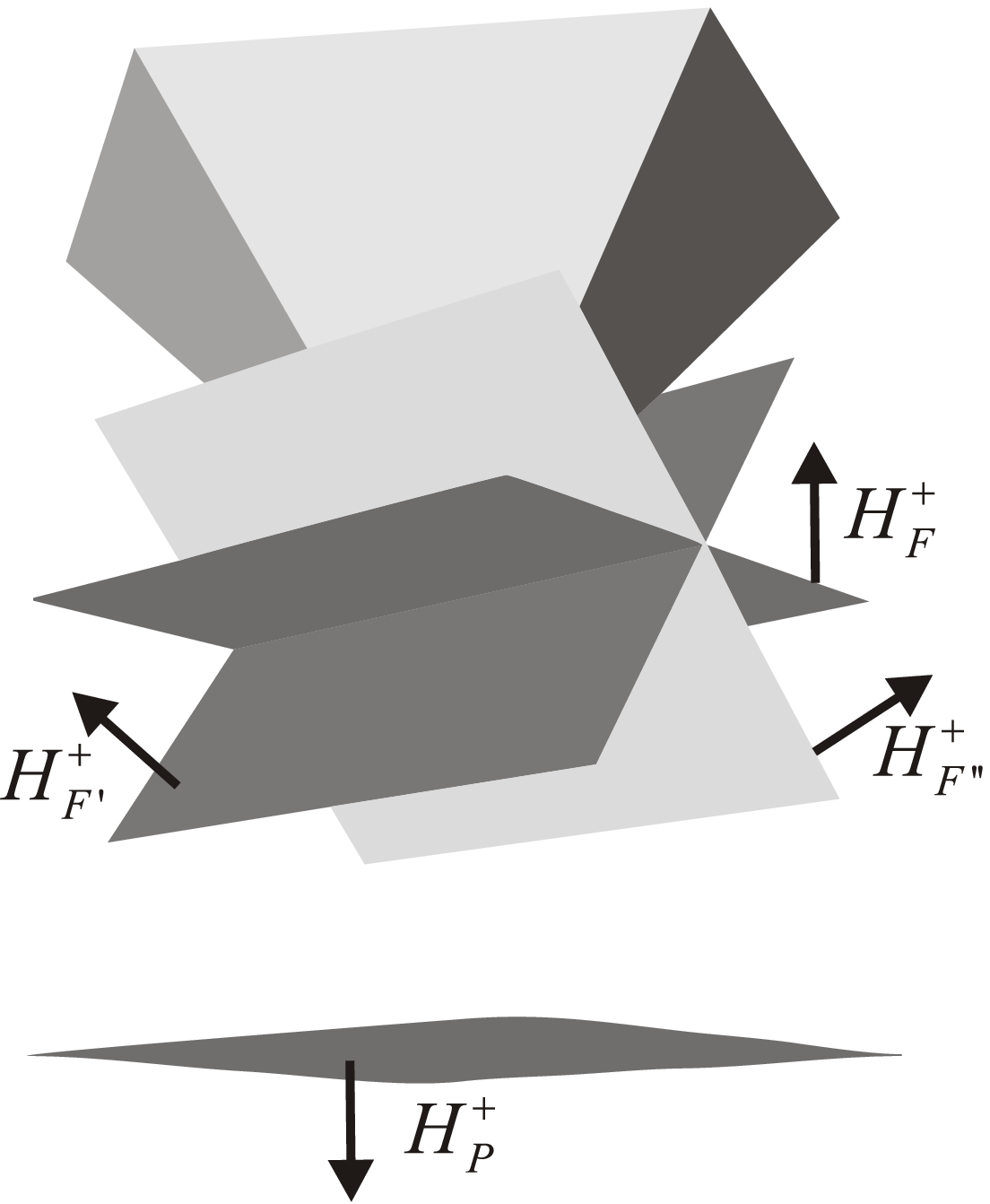}
\includegraphics[width=0.3\textwidth]{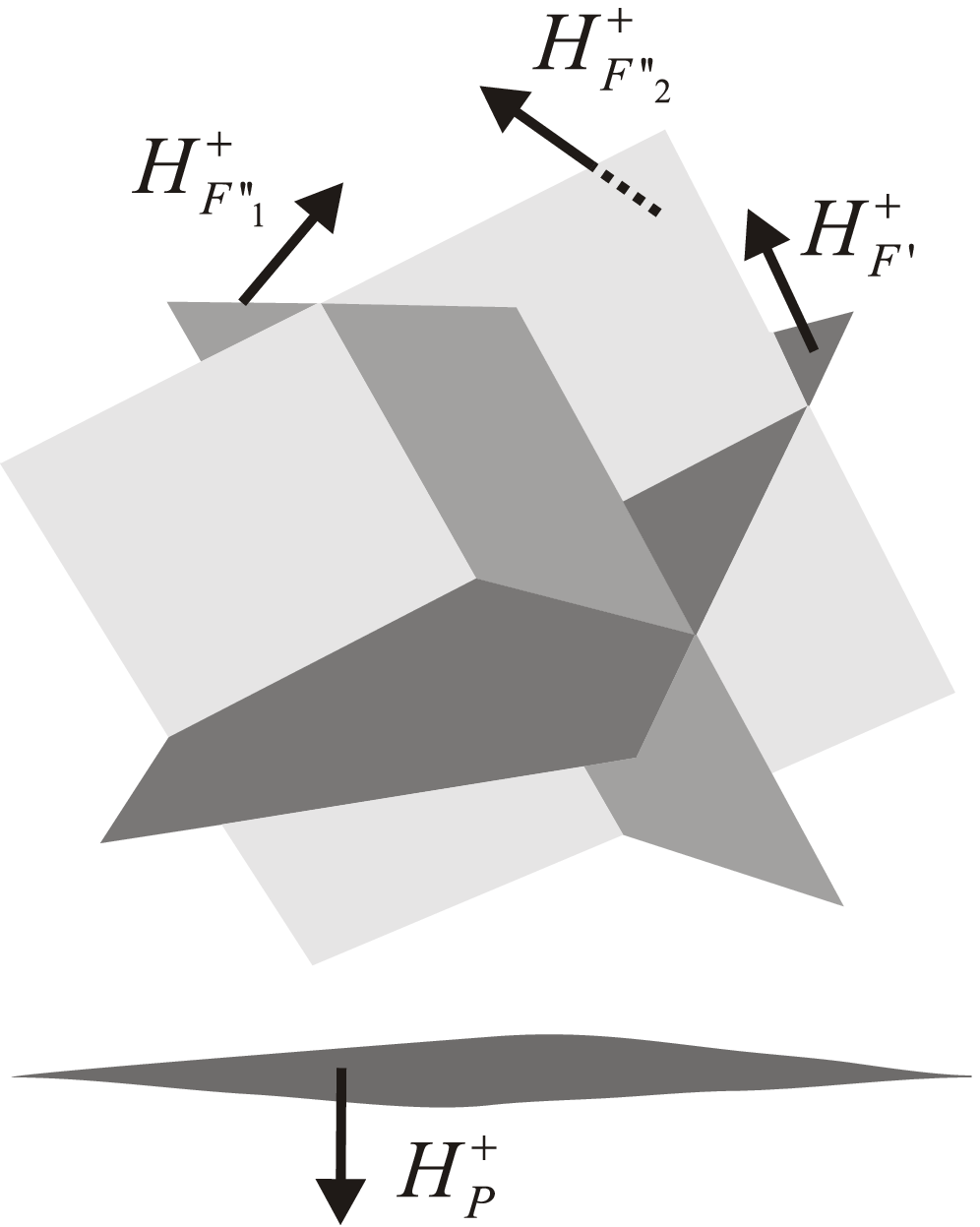}
\caption{Construction of a certificate of emptiness for the core.}\label{fig:certificate}
\end{figure}

\begin{lemma}\label{lemma:facce}
Let $P=\Pol(\mathcal{S})$ be a polyhedron of $\Re^n$ with $\dim(P)=n$, and $H_F^+$ a supporting halfspace of $P$ whose touching set is $F$. Then, there exists a set of halfspaces $\mathcal{H}_F\subseteq \mathcal{S}$ such that $|\mathcal{H}_F|\leq n-\dim(F)$, $H_F^+$ is a supporting halfspace of $\Pol(\mathcal{H}_F)$, and their touching set $C$ is such that
$F\subseteq C$.
\end{lemma}
\begin{proof}
The proof is by induction. \emph{Base case}:
If $\dim(F)=n-1$ we have that the touching face $F=H_F^+\cap P$ is a facet of $P$. Thus, from Fact~\ref{fact:generator}, $F$ is generated by some halfspace $H^+\subseteq \mathcal{S}$ such that $H^+\cap P=F$, as for $H_F$. Since $\dim(F)=\dim(H)=\dim(H_F)=n-1$, it easily follows that in fact $H=H_F$ holds. Thus, $H_F^+$ is trivially a supporting halfspace of $H^+$, and this case is proved: just take $\mathcal{H}_F =\{H^+\}$ and note that $|\mathcal{H}_F|=1$.

\emph{Inductive step}: By the induction hypothesis, the property holds for any supporting halfspace $H_{F'}^+$ of $P$ such that its touching face $F'$ has a dimension $d\leq\dim(F')\leq n-1$, for some $d>0$. We show that it also holds for any supporting halfspace $H_{F}^+$ of $P$, whose touching face $F$ has a dimension $\dim(F)=d-1$.

Since $F$ is not a facet, from Fact~\ref{fact:facet} there exists a facet $F'$ of $P$ such
that $F\subset F'$. Let $C=H_F\cap H_{F'}$, where $H_{F'}^+\in\mathcal{S}$ is the halfspace
that generates the facet $F'$. Note that $F\subseteq C$.
Let $H_{F''}^+$ be the halfspace that first touches the
polyhedron $P$ obtained by rotating $H_F^+$ towards $H_{F'}^+$ by some negative angle $\beta$ (i.e., we are going in the opposite direction w.r.t.\ $H_{F'}^+$). Then, $H_{F''}^+$ is a supporting halfspace of $P$, and the touching set $F''$
includes $C$. Since $\dim(P)=n$,
 we have that $\delta(H_F^+,H_{F'}^+)-\beta <\pi$ and  $F\subset F''$ (as the latter face contains some point of the polyhedron outside the axis $C\supseteq F$), and thus
the dimension $\dim(F'')$ of this face is strictly greater than $\dim(F)$, from Fact
\ref{fact:dimension}. Moreover,
 From Lemma~\ref{lemma:RoofLemma}, $H_F^+$ is a supporting halfspace of
 $\Pol(H_{F'}^+,H_{F''}^+)$, with $C=H_{F'}\cap H_{F''}$ as touching set.

 By the induction hypothesis, since both $\dim(F')$ and $\dim(F'')$ are at least $d$,
 we know that there are
 two sets $\mathcal{H}_{F'}\subseteq \mathcal{S}$ and $\mathcal{H}_{F''}\subseteq \mathcal{S}$ such that:
 $H_{F'}^+$ is a supporting halfspace of $\Pol(\mathcal{H}_{F'})$, with $F'\subseteq C'$,
 where $C'$ is their touching set; and
 $H_{F''}^+$ is a supporting halfspace of $\Pol(\mathcal{H}_{F''})$, with $F''\subseteq C''$,
 where $C''$ is their touching set.
  In particular, $\Pol(\mathcal{H}_{F'})\subseteq H_{F'}^+$ and
  $\Pol(\mathcal{H}_{F''})\subseteq H_{F''}^+$, respectively.
  Let $\mathcal{H}_F=\mathcal{H}_{F'}\cup\mathcal{H}_{F''}$.
  Then, we get $\Pol(\mathcal{H}_F)\subseteq \Pol(\{H_{F'}^+,H_{F''}^+\})\subseteq
 H_F^+$.

 Recall that $F\subseteq C$. Moreover, $F\subset F'$ and $F\subset F''$, and
 thus $F\subseteq (F'\cap F'') \subseteq (C'\cap C'')$, since $F'\subseteq C'$ and $F''\subseteq C''$.
 It follows that the touching set of  $H_F^+$ and
  $\Pol(\mathcal{H}_F)$ includes $F$.
 Indeed, this touching set includes
 $(C'\cap C'') \cap H_F \supseteq (C'\cap C'') \cap (H_F\cap H_{F'} \cap H_{F''}) = (C'\cap C'') \cap
 C$, and from the above observations $(C'\cap C'') \cap C \supseteq F$.

Finally, note that $|\mathcal{H}_F|\leq|\mathcal{H}_{F'}|+|\mathcal{H}_{F''}|=1+|\mathcal{H}_{F''}|$, because $\dim(F')=n-1$ and the base case applies.
 Moreover, $\dim(F'')>\dim(F)=d-1$ and thus, by the induction hypothesis, we obtain
 $|\mathcal{H}_F| \leq 1+n-dim(F'')\leq 1+n-d = n-\dim(F)$.
\end{proof}

Let $\game=\tuple{N,v}$ be a game with transferable payoffs. A coalition set $\mathcal{S}\subseteq 2^N$ is a
\emph{certificate of emptiness} (or \emph{infeasibility certificate}) for the core of $\game$ if the
intersection of $\Pol(\mathcal{S})$ with the grand\nbdash coalition halfspace
\eqref{eq:VincoloGrandCoalition} is empty.
In fact, this definition is motivated by the following observation. Let $P$ be the polyhedron of $\Re^n$
obtained as the intersection of all halfspaces \eqref{eq:VincoliCoalizioni}. Since $\mathcal{S}$ is a subset of
all possible coalitions, $P\subseteq\Pol(\mathcal{S})$. Therefore, if the intersection of $\Pol(\mathcal{S})$
with the grand\nbdash coalition halfspace \eqref{eq:VincoloGrandCoalition} is empty, the intersection of
this halfspace with $P$ is empty, as well.

\begin{theo}\label{theo:PrincipalTheorem}
Let $\game=\tuple{N,v}$ be a game with transferable payoffs. If the core of $\game$ is empty, there is a certificate of emptiness $\mathcal{S}$ for it such that $|\mathcal{S}|\leq |N|$.
\end{theo}
\begin{proof}
Let $n = |N|$ and $P$ be the polyhedron of $\Re^n$ obtained as the intersection of all halfspaces \eqref{eq:VincoliCoalizioni}. Since we are not considering the feasibility constraint \eqref{eq:VincoloGrandCoalition}, there is no upper-bound on the values of any variable $x_i$, and thus it is easy to see that $P\neq\emptyset$ and $\dim(P)=n$.

Let $H_P^+$ be the halfspace defined by the grand\nbdash coalition inequality \eqref{eq:VincoloGrandCoalition}. If the core of $\game$ is empty, the whole set of inequalities has no solution, that is, $P\cap H_P^+ =\emptyset$.

Let $\bar H_F^+$ be the halfspace parallel to $H_P^+$ that first touches $P$, that is, the smallest relaxation of $H_P^+$ that intersect $P$. Consider the opposite $H_F^+$ of $\bar H_F^+$, as shown in Figure~\ref{fig:certificate}, on the left. By construction, $H_P^+ \cap H_F^+=\emptyset$, $H_F = \bar H_F$ is a supporting hyperplane of $P$, and $H_F^+$ is a supporting halfspace of $P$. Let $F$ be the touching set of $H_F$ with $P$, and let $d=\dim(F)$. In  Figure~\ref{fig:certificate}, it is the vertex at the bottom of the diamond $P$. From Lemma~\ref{lemma:facce}, there is a set of halfspaces $\mathcal{S}$ associated with inequalities from \eqref{eq:VincoliCoalizioni}, with $|\mathcal{S}|\leq n-d$, and such that $H_F^+$ is a supporting halfspace for $\Pol(\mathcal{S})$. It follows that $H_P^+ \cap \Pol(\mathcal{S}) =\emptyset$, whence $\mathcal{S}$ is a certificate of emptiness for the core of $\game$. Finally, note that the largest cardinality of $\mathcal{S}$ is $n$, and corresponds to the case $\dim(F)=0$, that is, to the case where the face $F$ is just a vertex. Therefore the maximum cardinality of the certificate is $n$. In our three-dimensional example, the certificate is $\{H_{F'}^+,H_{F''_1}^+,H_{F''_2}^+\}$, as shown in Figure~\ref{fig:certificate}, on the right.
\end{proof}

Note that the above proof is constructive and has a nice geometrical interpretation.  Exploiting the above
property, we can now determine the complexity of core non\nbdash emptiness for any (non\nbdash deterministic)
polynomial\nbdash time compact representation.

\begin{theo}\label{theo:membership}
 Let $\R$ be a non-deterministic polynomial-time compact
representation. On the class $\C(\R)$, \CoreNonEmptyness\ is feasible in $\CONP$.
\end{theo}
\begin{proof}
Let $\game=\tuple{N,v}$ be a game with transferable payoffs in $\C(\R)$. If its core \emph{is} empty, from
Theorem~\ref{theo:PrincipalTheorem}, there is a certificate of emptiness $\mathcal{S}$, with
$|\mathcal{S}|\leq n$, where $n$ is the number of players of $\game$. For the sake of presentation, let us
briefly sketch the case of a polynomial\nobreakdash-time deterministic representation $\R$. In this case, a
non-deterministic Turing machine in polynomial time may check that the core is empty by performing the following
operations: (i) guess of the set $\mathcal{S}$, i.e., of the coalitions of players corresponding to the
halfspaces in $\mathcal{S}$; (ii) computation in deterministic polynomial time of the
worth $v^\mathcal{R}(\xi^\mathcal{R}(\game),S)$, for each
$S\in \mathcal{S}$, and for the grand\nbdash coalition $N$; and (iii) check of the property $\Pol(\mathcal{S})\cap
H_P^+=\emptyset$, where $H_P^+$ is the halfspace defined by the grand\nbdash coalition inequality
\eqref{eq:VincoloGrandCoalition}. Note that the last step is feasible in polynomial time, as we have to solve a
system consisting of just $n+1$ linear inequalities.

The case of a non-deterministic polynomial-time compact representation $\R$ is a simple variation where, at step (ii), for each $S\in \mathcal{S}$, the computation of the worth is an $\NPSV$ problem, and thus the machine should work as follows: it guesses the value $w=v^\mathcal{R}(\xi^\mathcal{R}(\game),S)$ and a witness $y$
 that $\tuple{(\xi^\mathcal{R}(\game),S),w}\in graph(v^\mathcal{R})$ (which is an $\NP$ task in this case); then, by exploiting $y$,
it checks in (deterministic) polynomial-time that actually $\tuple{(\xi^\mathcal{R}(\game),S),w}$ belongs to the graph of the worth function $v^\mathcal{R}$.
\end{proof}

The above membership result and Proposition~\ref{prop:hardness-core}
 imply that the non-emptiness problem for the core is \CONPc{} for all $\NP$ representations at least as expressive as the graph games.

\begin{corol}\label{cor:completezza-core}
Let $\R$ be any non-deterministic polynomial-time compact representation such that $\graph\precsim_e \R$ (e.g., $\R=\marginal$). On the class $\C(\R)$, \CoreNonEmptyness\ is \CONPc.
\end{corol}

\section{The Complexity of the Kernel}\label{sec:kernel}

The \emph{kernel} is a solution concept introduced by \citet{Davis_Maschler:Kernel}. Let us start by recalling
its formal definition.

For any pair of players $i$ and $j$ of a coalitional game $\game=\tugame$, we denote by $\mathcal{I}_{i,j}$ the
set of all coalitions containing player $i$ but not player $j$.
The \emph{excess} $e(S,x)=v(S)-x(S)$ of the generic coalition $S$ at the imputation $x\in X(\game)$ is a measure
of the dissatisfaction of $S$ at $x$.
Define the \emph{surplus} $s_{i,j}(x)$ of player $i$ against player $j$ at an imputation $x$ as the value
$s_{i,j}(x)=\max_{S\in\mathcal{I}_{i,j}}e(S,x)=\max_{S\in\mathcal{I}_{i,j}}(v(S)-x(S))$.

\begin{defin}\label{def:kernelTU}
The \emph{kernel} $\kernel{\game}$ of a TU game $\game=\tugame$ is the set:
\[
\kernel{\game}=\{x\in X(\game)\mid s_{i,j}(x)>s_{j,i}(x)\Rightarrow x_j=v(\{j\}),\forall i,j\in N,i\neq j\}.
\]

\vspace{-6mm}\hfill $\Box$
\end{defin}

Intuitively, the surplus of player $i$ against $j$ at $x$ is the highest payoff that player $i$ can gain (or the
minimal amount $i$ can lose, if it is a negative value) without the cooperation of $j$, by assuming to form
coalitions with other players that are satisfied at $x$; thus, $s_{i,j}(x)$ is the weight of a possible threat
of $i$ against $j$. In particular, player $i$ has more ``bargaining power'' than $j$ at $x$ if
$s_{i,j}(x)>s_{j,i}(x)$; however, player $j$ is immune to such threat whenever $x_j=v(\{j\})$, since in this
case $j$ can obtain $v(\{j\})$ even by operating on her own. We say that player $i$ outweighs player $j$ at $x$ if
$s_{i,j}(x)>s_{j,i}(x)$ and $x_j>v(\{j\})$. The kernel is then the set of all imputations where no player
outweighs another one.

\begin{example}\label{ex:K_TU}
Let $\game=\tugame$ be a TU game with $N=\{a,b,c\}$, $v(\{a\})=v(\{b\})=v(\{c\})=0$, $v(\{a,b\})=20$,
$v(\{a,c\})=30$, $v(\{b,c\})=40$, and $v(\{a,b,c\})=42$.

It is easily verified that the imputation $x$ such that $x_a=4$, $x_b=14$, and $x_c=24$ is in the kernel of
$\game$.
Indeed, we note first that every player in $N$ receives in $x$ a payoff strictly greater than what she is able
to obtain by acting on her own. For this reason, in order for $x$ to belong to $\kernel{\game}$ it must be the
case that $s_{i,j}(x)\leq s_{j,i}(x)$, for all distinct players $i$ and $j$. By the definition of the worth
function, the maximum excess that a coalition $S$ including $i$ and excluding $j$ can achieve is obtained by the
coalition $S\in\mathcal{I}_{i,j}$ such that $|S|=2$. By this, $s_{i,j}(x)=s_{j,i}(x)=2$ for all pairs of different players
$i,j$. Thus, $x\in\kernel{\game}$.~\hfill~$\lhd$
\end{example}

It is well known that $\kernel{\game}\neq\emptyset$, whenever $X(\game)\neq \emptyset$ (see, e.g.,
\cite{Osborne_Rubinstein:GameTheory}). Thus, the non-emptiness problem is trivial for this concept. Instead, as
discussed by~\citet{Deng_Papadimitriou:ComplexityCooperativeGameSolutionConcepts}, it is of interest to ask for
the computational complexity of deciding whether a given payoff vector belongs to the kernel. This problem was
conjectured to be $\NP$-hard by those same authors~\cite{Deng_Papadimitriou:ComplexityCooperativeGameSolutionConcepts}, even for
graph games.
In the rest of the section, we firstly confirm the conjecture by actually showing that the problem is even
harder and, precisely, \DeltaPh2. Then, we shall show that the corresponding membership result holds on any class of
games $\C(\R)$, where $\R$ is an arbitrary \NP{} representation scheme. Finally, we shall single out an island
of tractability, by looking for structural restrictions imposed over graph games.

\subsection{Hardness on Graph Games}\label{sec:hardnessKernelTU}

We show that checking whether an imputation belongs to the kernel is $\DeltaP{2}$-hard on graph games. The proof is based on a
reduction from the problem of the lexicographically maximum satisfying assignment for Boolean formulae.

\begin{figure}[t]
\centering
\includegraphics[width=0.8\textwidth]{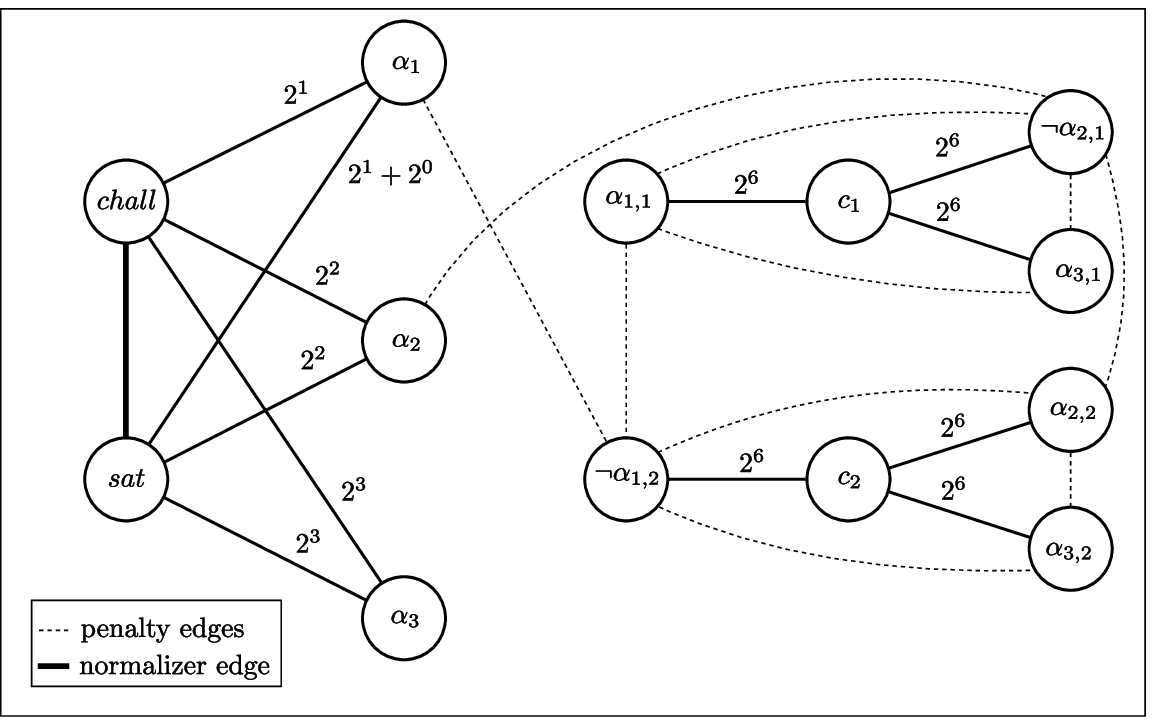}
\caption{The game $\qtugameKernel(\widehat{\phi})$, where $\widehat{\phi}= (\alpha_1\lor \lnot \alpha_2 \lor
\alpha_3) \land (\lnot \alpha_1\lor \alpha_2 \lor \alpha_3)$.}\label{fig:graphgameK}
\end{figure}

Let $\phi=c_1\wedge\dots\wedge c_m$ be a \textbf{3CNF} Boolean formula, that is, a Boolean formula in conjunctive normal form over the set
$\{\alpha_1,\dots,\alpha_n\}$ of variables that are lexicographically ordered (according to their indices) where each clause contains three
literals (positive or negated variables) at most.
Assume, w.l.o.g., that there is a clause in $\phi$ containing at least two literals. Based on $\phi$, we build
in polynomial time the weighted graph $\qtugameKernel(\phi)=\tuple{(N_\K,E_\K),w}$ such that (see
Figure~\ref{fig:graphgameK} for an illustration):

\begin{itemize}
\item The set $N_\K$ of nodes (i.e, players) includes: a \emph{variable player} $\alpha_i$, for each variable $\alpha_i$ in $\phi$; a
    \emph{clause player} $c_j$, for each clause $c_j$ in $\phi$; a \emph{literal player} $\ell_{i,j}$ (either $\ell_{i,j}=\alpha_{i,j}$ or
    $\ell_{i,j}=\neg \alpha_{i,j}$), for each literal $\ell_{i}$ ($\ell_{i}=\alpha_{i}$ or $\ell_i=\neg \alpha_{i}$, respectively)
    occurring in $c_j$; and, two special players ``$chall$'' and ``$sat$''.

\item The set $E_\K$ consists of the following three types of edges.
\begin{description}
\item[\textbf{(Positive edges):}] an edge $\{c_j,\ell_{i,j}\}$ with $w{(\{c_j,\ell_{i,j}\})}=2^{n+3}$, for each literal $\ell_i$ occurring in
    $c_j$; an edge $\{chall,\alpha_i\}$ with $w(\{chall,\alpha_i\})=2^i$, for each $1 \leq i \leq n$; an edge $\{sat,\alpha_i\}$ with
    $w(\{sat,\alpha_i\})=2^i$, for each $2 \leq i \leq n$; the edge $\{sat,\alpha_1\}$ with $w(\{sat,\alpha_1\})=2^1+2^0$.

  \item[\textbf{(``Penalty'' edges):}] an edge  $\{\ell_{i,j},\ell_{i',j}\}$ with $w{(\{\ell_{i,j},\ell_{i',j}\})}=-2^{m+n+7}$, for each pair of
      literals $\ell_i$ and $\ell_{i'}$ occurring in $c_j$; an edge $\{\alpha_{i,j},\neg \alpha_{i,j'}\}$ with $w(\{\alpha_{i,j},\linebreak[0]\neg
      \alpha_{i,j'}\})=-2^{m+n+7}$, for each variable $\alpha_i$ occurring positively in $c_j$ and negated in $c_{j'}$; an edge
      $\{\alpha_{i},\neg \alpha_{i,j}\}$ with $w(\{\alpha_{i},\neg \alpha_{i,j}\})=-2^{m+n+7}$, for each variable $\alpha_i$ occurring
      negated in $c_{j}$.

  \item[\textbf{(``Normalizer'' edge):}] the edge $\{chall,sat\}$, for which we assign the weight $w(\{chall,sat\})=$ $1-\sum_{e\in E_\K\mid e\neq\{chall,sat\}}w(e)$.
\end{description}
\end{itemize}

Note that the size of the representation of all the weights is polynomial in the number of variables and clauses of $\phi$. Two
crucial properties of the above construction are stated in the following lemma.

\begin{lemma}\label{lem:kernel}
Let $\qtugameKernel(\phi)=\tuple{(N_\K,E_\K),w}$ be the graph game associated with the \emph{\textbf{3CNF}}
formula $\phi$. Then:
\begin{itemize}
\item[{\em (A)}] $w(\{chall,sat\})\geq D+1$; and,

\item[{\em (B)}] $D+w(e)<0$, for each penalty edge $e\in E_\K$,
\end{itemize}
where $D=\max_{\{chall,sat\}\not\subseteq S\subseteq N} v(S)$ denotes the maximum worth over all the coalitions not covering the edge
$\{chall,sat\}$.
\end{lemma}
\begin{proof}
Let $P=\sum_{e\in E_\K\mid e\neq\{chall,sat\},w(e)>0}w(e)$ be the sum of all the positive edges, but the normalizer one, in
$\qtugameKernel(\phi)$.
Let us firstly observe that:
\[P\leq 3\times m\times 2^{n+3}+2\times\sum_{i=1}^n2^i+2^0\leq 2^{m+n+5}+2^{n+2}+2^0\leq  2^{m+n+6}.\]
Thus, $2^{m+n+7}\geq 2\times P$ holds. Moreover, observe that $\qtugameKernel(\phi)$ contains at least one penalty edge, since w.l.o.g.\ there
is a clause in $\phi$ containing at least two literals.
Hence, $w(\{chall,sat\})=1-\sum_{e\in E_\K\mid e\neq\{chall,sat\}}w(e)\geq 1-P+2^{m+n+7}$. It follows that $w(\{chall,sat\})\geq 1-P+2\times
P=1+P$. Eventually, $P\geq D$ holds by definition of $D$ and, therefore, $w(\{chall,sat\})\geq 1+D$, which proves {(A)}.

As for {(B)}, given that $D>0$ and $P\geq D$, we may note that $2^{m+n+7}\geq 2\times P$ implies $2^{m+n+7}>D$.
\end{proof}

Based on the above properties, we can now prove the main result.

\begin{theo}\label{theo:hardnessKernelCheckingTU_Graph}
On the class $\C(\graph)$ of graph games, \KernelMembership\ is \DeltaPh2 (even if the given payoff vector is an
imputation).
\end{theo}
\begin{proof}
Let $\phi=c_1\wedge\dots\wedge c_m$ be a satisfiable \textbf{3CNF} formula over a set
$\{\alpha_1,\dots,\alpha_n\}$ of variables that are lexicographically ordered (according to their indices).
Deciding whether $\alpha_1$ (that is the lexicographically least significant variable) is true in the
\emph{lexicographically maximum satisfying assignment} for $\phi$ is a well-known \DeltaPc2{}
problem~\citep{Krentel:ComplexityOptimizationProblems}.

Consider the graph game $\qtugameKernel(\phi)=\tuple{(N_\K,E_\K),w}$, and the imputation $x$ that assigns $0$ to
all players of $\qtugameKernel(\phi)$, but to $sat$, which receives $1$. In fact, $x$ is an imputation since
$v(N)=1$ because of the weight of the edge $\{chall,sat\}$ (recall that $w(\{chall,sat\})=1-\sum_{e\in E_\K\mid
e\neq\{chall,sat\}}w(e)$).

Beforehand, observe that, by Definition~\ref{def:kernelTU} and since $sat$ is the only player that receives in $x$ a payoff strictly greater
than her worth as a singleton coalition, $x\in\kernel{\qtugameKernel(\phi)}$ if and only if $\max_{S\in\mathcal{I}_{i,sat}}e(S,x)\leq
\max_{S\in\mathcal{I}_{sat,i}}e(S,x)$, for each player $i\neq sat$. The following property further restricts the coalitions of interest and, indeed, suffices to conclude that $x\in\kernel{\qtugameKernel(\phi)}$ if and only if
\begin{equation}\label{eq:kernel}
\max_{S\in\mathcal{I}_{chall,sat}}e(S,x)\leq \max_{S\in\mathcal{I}_{sat,chall}}e(S,x).
\end{equation}

\begin{description}
\item[\textbf{Property~\ref{theo:hardnessKernelCheckingTU_Graph}.(1):}] \emph{For each player $i\not\in\{sat,chall\}$, 
    $\max_{S\in\mathcal{I}_{i,sat}}e(S,x)\leq\max_{S\in\mathcal{I}_{sat,i}}e(S,x)$.}\\
Let $S$ be an arbitrary coalition in $\mathcal{I}_{i,sat}$
with $i\neq chall$,
and consider the coalition $T=\{chall,sat\}\in \mathcal{I}_{sat,i}$. Note that
$e(T,x)=v(T)-x(T)=v(\{chall,sat\})-1$ while $e(S,x)=v(S)$. By Lemma~\ref{lem:kernel}.(A), we know that $v(\{chall,sat\})\geq v(S)+1$. Thus, $e(T,x)\geq e(S,x)$ holds
for all coalitions $S\in\mathcal{I}_{i,sat}$ with $i\neq chall$.
\end{description}

Now, we are going to characterize the structure of the two terms $\max_{S\in\mathcal{I}_{chall,sat}}e(S,x)$ and
$\max_{S\in\mathcal{I}_{sat,chall}}e(S,x)$ occurring in Equation~\eqref{eq:kernel} by establishing a connection
with satisfying assignments for $\phi$. In particular, for any truth assignment $\sigma$, we denote by
$\sigma\models \phi$ the fact that $\sigma$ satisfies $\phi$, and by $\sigma(\alpha_i)=\mathtt{true}$ (resp.,
$\sigma(\alpha_i)=\mathtt{false}$) the fact that $\alpha_i$ evaluates to true (resp., false) in $\sigma$.

\begin{description}
\item[\textbf{Property~\ref{theo:hardnessKernelCheckingTU_Graph}.(2):}] $\max_{S\in\mathcal{I}_{chall,sat}}e(S,x)=m\times
    2^{n+3}+\max_{\sigma\models \phi} \sum_{\alpha_i\mid \sigma(\alpha_i)=\mathtt{true}} 2^i.$\\
Let us firstly note that $\max_{S\in\mathcal{I}_{chall,sat}}e(S,x)=\max_{S\in\mathcal{I}_{chall,sat}}v(S)$, by construction of the
imputation $x$. Let $S_*$ be the coalition having maximum worth over all the coalitions in $\mathcal{I}_{chall,sat}$. Because of
Lemma~\ref{lem:kernel}.(B) and since $v(\{chall\})=0$, $S_*$ cannot cover any penalty edge, for otherwise $S_*$ would not be a coalition
with maximum worth amongst those belonging to $\mathcal{I}_{chall,sat}$. Thus, \emph{(i)} for each variable player $\alpha_i\in S_*$, no
literal player of the form $\neg \alpha_{i,j}$ is in $S_*$; \emph{(ii)} for each clause player $c_j\in S_*$, at most one literal player of
the form $\ell_{i,j}$ is in $S_*$; and, \emph{(iii)} for each variable $\alpha_i$, $S_*$ contains no pair of literal players of the form
$\alpha_{i,j}$ and $\neg \alpha_{i,j'}$.

It follows that the worth of $S_*$ is such that: $v(S_*)=|C|\times 2^{n+3}+\sum_{\alpha_i\in S_*} 2^i$, where $C$ is the set of the clause
players $c_j\in S_*$ for which exactly one literal player $\ell_{i,j}$ is in $S_*$; in particular, recall that $2^i$ is the weight
associated with the edge $\{chall,\alpha_i\}$, while $2^{n+3}$ is the weight associated with each edge of the form $\{c_j,\ell_{i,j}\}$.
Now, let, $\widehat{\sigma}$ be a truth assignment such that $\widehat{\sigma}(\alpha_i)=\mathtt{true}$ (resp.,
$\widehat{\sigma}(\alpha_i)=\mathtt{false}$) if $\alpha_{i,j}$ (resp., $\neg \alpha_{i,j}$) occurs in $S_*$ for some clause $c_j$. Note
that $\widehat{\sigma}$ may be a partial assignment, over a set of variables $\widehat{\alpha}\subseteq\{\alpha_1,\dots,\alpha_n\}$;
however, because of \emph{(iii)} above, $\widehat{\sigma}$ is non\nbdash contradictory and satisfies all the clauses whose players are in $C$. Eventually,
since $\phi$ is satisfiable and since $2^{n+3}>\sum_{i=1}^n 2^i$, because of \emph{(ii)} $S_*$ will certainly contain all the $m$ clause
players (i.e., $\widehat{\sigma}$ is a satisfying assignment for $\phi$). That is, $v(S_*)=m\times 2^{n+3}+\sum_{\alpha_i\in S_*} 2^i$.

Observe now that if $\alpha_{i,j}$ is in $S_*$, then $\alpha_i$ is in $S_*$ as well, since this leads to
maximize the worth of $S_*$. Moreover, if $\neg \alpha_{i,j}$ is in $S_*$, then $\alpha_i$ is not in $S_*$
because of \emph{(i)}. Thus, the assignment $\sigma_{S_*}$ such that $\sigma_{S_*}(\alpha_i)=\mathtt{true}$
(resp., $\sigma_{S_*}(\alpha_i)=\mathtt{false}$) if $\alpha_i$ occurs (resp., not occurs) in $S_*$ coincides
with $\widehat{\sigma}$ when restricted over the domain of the variables in $\widehat{\alpha}$. Therefore,
$\sigma_{S_*}$ is, in turn, a satisfying assignment, and we have:

\[v(S_*)= m\times 2^{n+3}+\sum_{\alpha_i\mid \sigma_{S_*}(\alpha_i)=\mathtt{true}} 2^i\leq m\times
2^{n+3}+\max_{\sigma\models \phi} \sum_{\alpha_i\mid \sigma(\alpha_i)=\mathtt{true}} 2^i.\]

We conclude the proof by showing that the above inequality cannot be strict. Indeed, assume, for the sake of contradiction, that a
satisfying assignment $\overline{\sigma}$ exists for $\phi$ such that $v(S_*)<m\times 2^{n+3}+\sum_{\alpha_i\mid
\overline{\sigma}(\alpha_i)=\mathtt{true}} 2^i$. Based on $\overline{\sigma}$, we can build a coalition $\overline{S}$ such that: \emph{(a)}
$\{chall,c_1,\dots,c_m\}\subseteq \overline{S}$; \emph{(b)} $\alpha_i\in \overline{S}$, for each $\alpha_i$ such that
$\overline{\sigma}(\alpha_i)=\mathtt{true}$; \emph{(c)} exactly one literal $\ell_{i,j}$ is in $\overline{S}$, for each clause $c_j$ that is satisfied
by $\ell_{i,j}$ according to the truth values defined in $\overline{\sigma}$; \emph{(d)} no further player is in $\overline{S}$.

Given that $\overline{\sigma}$ is a satisfying assignment, no penalty edge is covered by $\overline{S}$. In particular,
$v(\overline{S})=m\times 2^{n+3}+\sum_{\alpha_i\in \overline{S}} 2^i$ and, hence, $v(\overline{S})=m\times 2^{n+3}+\sum_{\alpha_i\mid
\overline{\sigma}(\alpha_i)=\mathtt{true}} 2^i$. But, this is not possible since we would have a coalition
$\overline{S}\in\mathcal{I}_{chall,sat}$ such that $v(\overline{S})>v(S_*)=\max_{S\in\mathcal{I}_{chall,sat}}v(S)$.

\item[\textbf{Property~\ref{theo:hardnessKernelCheckingTU_Graph}.(3):}] {\small
$\max_{S\in\mathcal{I}_{sat,chall}}e(S,x)= m\times 2^{n+3}+\max_{\sigma\models \phi} (|{ \{\alpha_1\mid
\sigma(\alpha_1)}=\mathtt{true}\}|+\sum_{\alpha_i\mid
    \sigma(\alpha_i)=\mathtt{true}} 2^i)-1.$}\\

\vspace{-2mm}The property can be proven precisely along the same line of reasoning as in the proof of {Property
\ref{theo:hardnessKernelCheckingTU_Graph}.(2)}. The differences are that: $x(S)=1$ holds for each $S$ with
$sat\in S$; and that the weight associated with the edge $\{sat,\alpha_i\}$ is $2^i$, for each $2\leq i\leq n$,
while it is $2^1+2^0$ for the case where $i=1$. In particular, $|\{\alpha_1\mid \sigma(\alpha_1)=\mathtt{true}\}|$
precisely encodes the fact that an unitary weight has to be added to any assignment where $\alpha_1$ evaluates
to true.
\end{description}

\noindent We can now rewrite Equation~\eqref{eq:kernel} in the light of the above two properties, and conclude that
$x\in\kernel{\qtugameKernel(\phi)}$ if and only if:
\[1+\max_{\sigma\models \phi} \sum_{\alpha_i\mid \sigma(\alpha_i)=\mathtt{true}} 2^i\leq
\max_{\sigma\models \phi} \left(\sum_{\alpha_i\mid \sigma(\alpha_i)=\mathtt{true}} 2^i+|\{\alpha_1\mid
\sigma(\alpha_1)=\mathtt{true}\}|\right),\]

\noindent that is, $x\in\kernel{\qtugameKernel(\phi)}$ if and only if $\alpha_1$ is true in the  {lexicographically maximum satisfying
assignment} for $\phi$.
\end{proof}

By the same argument used in the proof of Proposition~\ref{prop:hardness-core}, the above result can be immediately extended to all representations at least as expressive (and succinct) as graph games.

\begin{corol}\label{cor:hardness-kernel}
Let $\R$ be any compact representation such that $\graph\precsim_e \R$ (e.g., $\R=\marginal$). On the class $\C(\R)$, \KernelMembership\ is \DeltaPh2.
\end{corol}

\smallskip

\subsection{Membership on FNP Representation Schemes}

Differently from the hardness result proven above, the corresponding membership result is routine.
 Anyway, the
proof is reported below, for the sake of completeness.

\begin{theo}\label{theo:membershipKernelCheckingTU}
Let $\mathcal{R}$ be a non-deterministic polynomial\nbdash time compact representation. On the class $\C(\R)$,
\KernelMembership\ is feasible in \DeltaP2.
\end{theo}
\begin{proof}
Let $\game\in\C(\mathcal{R})$ be a coalitional game, and $x$ be a payoff vector. Note first that, in order to check whether $x$ is actually an imputation, we have just to compute the worths associated to all the singleton coalitions plus the worth associated to the whole set of players, which can be done in \DeltaP2---just recall that each worth can be computed in non-deterministic polynomial\nbdash time. More in detail, for each player $i\in \{1,\dots,n\}$, we guess in polynomial time the value $w_i=v^\mathcal{R}(\xi^\mathcal{R}(\game),\{i\})$ and a suitable certificate $c_{\{i\}}$, and then we check in (deterministic) polynomial time that in fact $\tuple{(\xi^\mathcal{R}(\game),\{i\}),w_i}\in graph(v^\mathcal{R})$, exploiting the certificate $c_{\{i\}}$ (see the proof of Theorem~\ref{theo:membership}).

Assume now that $x$ is an
imputation. Then, we observe that for each pair of players $i$ and $j$, we can compute the value $s_{i,j}(x)$ by
means of a binary search over the range of the possible values for the worth functions, by using an \NP{}
oracle. Indeed, for any value $h$ in this range, we can decide in $\NP$ whether there is a coalition $S\in \mathcal{I}_{i,j}$ such
that $e(S,x)>h$, by guessing the coalition $S$, its worth $w=v^\mathcal{R}(\xi^\mathcal{R}(\game),S)$ together with a (polynomial) certificate $c_S$, and then by checking in polynomial time that in fact
$S\in \mathcal{I}_{i,j}$, that
 $\tuple{(\xi^\mathcal{R}(\game),S),w}\in graph(v)$ (by exploiting $c_S$), and that $e(S,x)>h$, where $e(S,x)=w-x(S)$.

Finally, note that, being $\R$ a non\nbdash deterministic polynomial\nbdash time compact representation, $v^\mathcal{R}(\cdot,\cdot)$ is polynomially balanced, meaning that there exists a fixed integer $k$ such that, for every game $\game$ in $\C(\mathcal{R})$, $||v^\mathcal{R}(\xi^\mathcal{R}(\game),S)||\leq||\tuple{\xi^\mathcal{R}(\game),S}||^k$ for every coalition $S$, and hence the worth value of every coalition in $\game$ (of every $\game\in\C(\R)$) cannot exceed the largest (positive or negative) value representable in the polynomial size $||\tuple{\xi^\mathcal{R}(\game),S}||^k$.
By this, the above binary search allows us to find the maximum excess in at most polynomially many steps.
Therefore, by using polynomially\nbdash many oracle
calls, we may check in polynomial time that for each pair of distinct players $i$ and $j$ such that $x_j\neq v(\{j\})$, it is the case that
$s_{i,j}(x)\leq s_{j,i}(x)$.
\end{proof}

\smallskip

From this result and Corollary~\ref{cor:hardness-kernel}, we immediately get the following completeness result.

\begin{corol}
Let $\R$ be any non-deterministic polynomial\nbdash time compact representation such that $\graph\precsim_e \R$ (e.g., $\R=\marginal$). On the class $\C(\R)$, \KernelMembership\ is \DeltaPc2.
\end{corol}

\subsection{Tractable Classes of Graph Games}

Many $\NP$-hard problems in different application areas ranging, e.g., from AI~\cite{pear-jeav-97} and Database
Theory~\cite{bern-good-81} to Game theory~\cite{DP06}, are known to be efficiently solvable when restricted to
instances whose underlying structures can be modeled via \emph{acyclic} graphs or \emph{nearly-acyclic} ones,
such as those graphs having \emph{bounded treewidth}~\cite{RS84}. Indeed, on these kinds of instances, solutions
can usually be computed via dynamic programming, by incrementally processing the acyclic (hyper)graph, according
to some of its topological orderings.
In this section, we shall show that this is also the case for the kernel of bounded treewidth graph games.

Our result is established by showing that this concept can be expressed in terms of an optimization problem over
\emph{Monadic Second Order Logic (MSO)} formulas, and by subsequently applying Courcelle's Theorem
\cite{Courcelle90} and its generalization to optimization problems due to Arnborg, Lagergren, and
Seese~\cite{ALS91}. Thus, we first review the concepts of treewidth and MSO, and their relationship as it
appears from Courcelle's Theorem.

\medskip

\noindent \textbf{Treewidth.} A \emph{tree decomposition} of a graph $G=(N,E)$ is a pair $\tuple{T,\chi}$, where
$T=(V,F)$ is a tree, and $\chi$ is a labeling function assigning to each vertex $p\in V$ a set of vertices
$\chi(p)\subseteq N$, such that the following conditions are satisfied:
\begin{enumerate}
\item[(1)] for each node $b$ of $G$, there exists $p\in V$ such that $b\in \chi(p)$;

\item[(2)] for each edge  $(b,d)\in E$, there exists $p\in V$ such that $\{b,d\}\subseteq \chi(p)$; and,

\item[(3)] for each node $b$ of $G$, the set $\{p\in V \mid  b\in \chi(p)\}$ induces a connected subtree.
\end{enumerate}

The \emph{width} of $\tuple{T,\chi}$ is the number $\max_{p\in V}(|\chi(p)|-1)$. The {\em treewidth} of $G$, denoted by $tw(G)$, is the minimum
width over all its tree decompositions. It is well-known that a graph $G$ is acyclic if and only if $tw(G)=1$.
Deciding whether a given graph has treewidth bounded by a fixed natural number $k$ is known to be feasible in linear time~\citep{bodl-96}.

\medskip
\noindent \textbf{MSO.} Monadic Second Order (MSO) logic formulae over graphs are made up of the logical connectors $\vee$, $\wedge$, and $\neg$,
the membership relation $\in$, the quantifiers $\exists$ and $\forall$, vertices variables and vertex sets variables---in addition, it is
often convenient to use symbols like $\subseteq$, $\subset$, $\cap$, $\cup$, and $\rightarrow$ with their usual meaning, as abbreviations.

\citet{Courcelle90} and other authors considered an extension of MSO, called MSO$_2$, where variables for edge sets are also allowed.
The fact that an MSO$_2$ sentence $\phi$ holds over a graph $G$ is denoted by $G\models \phi$.

The relationship between treewidth and MSO$_2$ is illustrated next. For a positive constant $k$, let hereafter $\mathcal{C}_k$ be a class of
graphs having treewidth bounded by $k$.
For a graph $G=(N,E)$, we denote by $||G||$ its \emph{size}, which is measured as the number of its nodes plus the number of its edges, i.e.,
$||G||=|N|+|E|$

\begin{prop}[\citeauthor{Courcelle90}'s Theorem~\cite{Courcelle90}]
Let $\phi$ be a fixed {\em MSO$_2$} sentence.
For each $G\in \mathcal{C}_k$, deciding whether $G\models \phi$ or not is feasible in linear time (w.r.t.\ $||G||$).
\end{prop}

An important generalization of MSO$_2$ formulae to \emph{optimization} problems was presented by~\citet{ALS91}. Next, we
state a simplified version of these kinds of problems.
Optimization problems are defined over MSO$_2$ formulae containing free variables and over graphs that are weighted on both nodes and edges.

Let $G=\tuple{(N,E),f_N,f_E}$ be a weighted graph where $f_N$ and $f_E$ are the lists of weights associated with
nodes and edges, respectively. Then, $f_N(v)$ (resp., $f_E(e))$ denotes the weight associated with $v\in N$ (resp., $e\in E$).
For the weighted graph $G=\tuple{(N,E),f_N,f_E}$, its size $||G||$ is measured as the sum of $|N|+|E|$ plus all values in the lists $f_N$ and $f_E$. Note that this is equivalent to state that weights are encoded in unary notation and that $||G||$ is measured as the sum of
$|N|+|E|$ plus all the bits that that are necessary to store the lists $f_N$ and $f_E$.

Let $\phi(X,Y)$ be an MSO$_2$ formula over the graph $(N,E)$, where $X$ and $Y$ are the free variables occurring in $\phi$, with $X$ (resp.,
$Y$) being a vertex (resp., edge) set variable. For a pair of interpretations $\tuple{z_N,z_E}$ mapping $X$ to subsets of $N$ and $Y$ to subsets of
$E$, we denote by $\phi[\tuple{z_N,z_E}]$ the MSO$_2$ formula (without free variables) where $X$ and $Y$ are replaced by the sets $z_N(X)$ and $z_E(Y)$, respectively.

A \emph{solution} to $\phi$ over $G$ is a pair of interpretations $\tuple{z_N,z_E}$ such that $(N,E)\models \phi[\tuple{z_N,z_E}]$ holds.
The \emph{cost} of $\tuple{z_N,z_E}$ is the value  $\sum_{x\in z_N(X)}f_N(x)+\sum_{y\in z_E(Y)}f_E(y)$. A solution of minimum cost is said
\emph{optimal}.

\begin{theo}[simplified from~\citet{ALS91}]\label{theo:mainTracatbility}
Let $\phi$ be a fixed {\em MSO$_2$} sentence. For each weighted graph $G=\tuple{(N,E),f_N,f_E}$ such that $(N,E)\in \mathcal{C}_k$,
computing an optimal solution to $\phi$ over $G$ is feasible in (deterministic) polynomial time (w.r.t. $||G||$).
\end{theo}

\medskip

\noindent \textbf{A Tractable Class.} Now that we have discussed the preliminary notions and concepts related to
treewidth and MSO, we are in the position of stating our tractability result.

\begin{theo}
Let $\game=\tuple{(N,E),w}$ be a graph game such that $(N,E)\in \mathcal{C}_k$, and let $x$ be a payoff vector. Then, deciding whether $x\in
\kernel{\game}$ is feasible in polynomial time (w.r.t. $||G||$).
\end{theo}
\begin{proof}
Since we can check in polynomial time whether $x$ is an imputation, by Definition~\ref{def:kernelTU} we have just to show that checking whether the condition $s_{i,j}(x)>s_{j,i}(x)\Rightarrow x_j=v(\{j\})$ holds is feasible in polynomial-time. We prove that
 $s_{i,j}(x)$ may be computed in polynomial time for each each pair of players $i\neq j$, which clearly entails the above tractability result.

For each $X\subseteq N$ and
$Y\subseteq E$, consider the following MSO$_2$ formula, stating that $Y$ is the set of all those edges $e\in E$ such that $e\subseteq X$:
\begin{multline*}
proj(X,Y)\equiv
 \forall v,v'\Bigl( \{v,v'\}\in Y \rightarrow \{v,v'\}\subseteq X \Bigr) \wedge\\
\forall v,v'\Bigl( \{v,v'\}\subseteq X \wedge \{v,v'\}\in E \rightarrow \{v,v'\}\in Y\Bigr).
\end{multline*}
Let $w_E$ and $w_N$ be such that $w_E(\{v,v'\})=-w(\{v,v'\})$ and $w_N(v)=x_v$, and observe that  $ \max_{S\subseteq N} e(S,x)= -1 * opt$, where
$opt=\min_{S\subseteq N}(x(S)-v(S))$ is the cost of an optimal solution to $proj(X,Y)$ over $\tuple{(N,E),w_N,w_E}$.

Recall, now, that $s_{i,j}(x)=\max_{S\in\mathcal{I}_{i,j}}e(S,x)$.
Therefore, in order to exploit the above MSO formula for its computation, we have to make sure that the set of considered coalitions is restricted to $\mathcal{I}_{i,j}$.
We thus modify the weights of $i$ and $j$ in $\game$ so that, in any optimal solution of the above
formula, $i\in X$ and $j \notin X$, and thus $\max_{S\subseteq N} e(S,x)$ coincides with $s_{i,j}$.
 Let $B=\max_{e\in E} {\it abs}(w_E(e))$ be the largest weight (in absolute value) among all edges of the graph.
 Then, replace the weights for nodes $i$ and $j$ as follows: $w_N(i)=-1- |\{e\in E \mid i\in e\}|*B $ and  $w_N(j)= 1+ |\{e\in E \mid j\in e\}|*B$. It is straightforward to check that any optimal solution for this modified graph is attained at some set $S$ that includes $i$ and does not include $j$. From its value, say $opt'$, we may easily compute the desired value $s_{i,j}(x)$ as $-1 * (opt' + x_i - w_N(i))$.

Hence, by Theorem~\ref{theo:mainTracatbility} and the above MSO$_2$ formula, $s_{i,j}$ is computable in polynomial time, and thus the membership of $x$ in $\kernel{\game}$ can be decided in polynomial time, too.
\end{proof}

\section{The Complexity of the Bargaining Set}\label{sec:bargainingSet}

The concept of bargaining set was defined by \citet{Auman_Maschler:BaragingSet_TUGames} (see also \citet{Maschler:SurveyHandbook}). We start by recalling its formal definition.

Let $\game=\tuple{N,v}$ be a coalitional game, and $x\in X(\game)$ be an imputation. Let $S\subseteq N$ be a
coalition, and $y$ be an $S$\nbdash feasible payoff vector (i.e., $y(S)=v(S))$.
The pair $(y,S)$ is an \emph{objection of player $i$ against player $j$ to $x$} if $i\in S$, $j\notin S$, and
$y_k > x_k$ for all $k\in S$.

A \emph{counterobjection to the objection $(y,S)$ of $i$ against $j$ to $x$} is a pair $(z,T)$ where $j\in T$,
$i\notin T$, and $z$ is a $T$\nbdash feasible payoff vector such that $z_k \geq x_k$ for all $k\in T\setminus S$
and $z_k \geq y_k$ for all $k\in T\cap S$. If there does not exist any counterobjection to $(y,S)$, we say that
$(y,S)$ is a \emph{justified objection}.

\begin{defin}\label{def:bargSetTU}
The \emph{bargaining set} $\bargset{\game}$ of a TU game $\game$ is the set of all imputations $x$ to which
there is no justified objection. \hfill $\Box$
\end{defin}

\begin{example}\label{ex:BS_TU}
Let $\game=\tugame$ be a TU game with $N=\{a,b,c\}$, $v(\{a\})=v(\{b\})=v(\{c\})=0$, $v(\{a,b\})=20$,
$v(\{a,c\})=30$, $v(\{b,c\})=40$, and $v(\{a,b,c\})=42$.

Consider the imputation $x$ such that $x_a=8$, $x_b=10$, and $x_c=24$. An objection of player $c$ against player
$a$ to $x$ is $((12,28),\{b,c\})$. Player $a$ can counterobject to this objection using $((8,12),\{a,b\})$.
Another objection of player $c$ against player $a$ to $x$ is $((14,\linebreak[0]26),\linebreak[0]\{b,c\})$. In
this case, player $a$ cannot counterobject. The reason is that coalition $\{a,b\}$ receives a payoff $20$ and
this is not sufficient for player $a$ to counterobject since she needs at least $8$ for herself and at least
$14$ for player $b$, in order to respond to the proposal of player $c$. Therefore the imputation $x$ does not belong
to $\bargset{\game}$. The intuitive reason is that player $a$ receives too much, according to this profile.

Consider now the imputation $x'$ such that $x'_a=4$, $x'_b=14$, and $x'_c=24$. We focus on the objections of
player $a$ against player $c$. We note that, in order to object, player $a$ has to form the coalition $S=\{a,b\}$.
The excess $e(S,x)$ of $S$ at $x$ is $2$, hence players $a$ and $b$ have the possibility to distribute among
themselves a payoff of $2$ to make the objection. But player $c$ can always counterobject to player $a$ because
she can form the coalition $T=\{b,c\}$ whose excess at $x$ is $2$ and hence she can always match the proposal
made to player $b$ by player $a$ in order to object. A similar argumentation holds for every objection of every
player against any other. Thus $x'\in\bargset{\game}$.~\hfill~$\lhd$
\end{example}

It is well-known that $\kernel{\game}\subseteq \bargset{\game}$ (hence, $\bargset{\game}\neq\emptyset$, whenever
$X(\game)\neq \emptyset$), and that $\core{\game}\subseteq \bargset{\game}$ (see, e.g.,
\cite{Osborne_Rubinstein:GameTheory}). Thus, as in the case of the kernel, we shall just focus on the complexity
of deciding whether a given payoff vector is in the bargaining set. The problem was conjectured to be
\PiP2-complete for graph games by \citet{Deng_Papadimitriou:ComplexityCooperativeGameSolutionConcepts}. In this
section, we show that the conjecture is indeed correct. Also, we are able to generalize the result by showing that membership in \PiP2 holds on any class of
games $\C(\R)$, where $\R$ is an \NP{} representation scheme.

\subsection{Hardness on Graph Games}\label{sec:hardnessBargSetTU}

It has been suggested by \citet{Maschler:SurveyHandbook} that {computing} the bargaining set might be
intrinsically more complex than computing the core. The result presented in this section provides some fresh evidence that
this is indeed the case inasmuch as we show that the checking whether a payoff vector is in the bargaining set is \PiPh2,
even over graph games. The reduction is from the validity of quantified Boolean formulae.

Let $\bs{\Phi}=(\forall\bs{\alpha})(\exists\bs{\beta})\phi(\bs{\alpha},\bs{\beta})$ be an $\mathbf{NQBF}_{2,\forall}$ formula, i.e., a
quantified Boolean formula over the variables ${\bs \alpha}=\{\alpha_1,\dots,\alpha_n\}$ and ${\bs \beta}=\{\beta_1,\dots,\beta_r\}$, where
$\phi(\bs{\alpha},\bs{\beta})=c_1\wedge\dots\wedge c_m$ is a \textbf{3CNF} formula, and where each universally quantified variable $\alpha_k\in
\bs \alpha$ occurs only in the two clauses $c_{i(k)}=(\alpha_k\lor\lnot\beta_k)$ and $c_{\bar i(k)}=(\lnot\alpha_k\lor\beta_k)$---intuitively,
each variable $\alpha_k$ enforces the truth-value of a corresponding variable $\beta_k$, which thus plays its role in the formula
$\phi(\bs{\alpha},\bs{\beta})$.
Based on $\bs{\Phi}$, we define the weighted graph $\qtugameBargSet(\bs{\Phi})=\tuple{(N_\BS,E_\BS),w}$ such
that (see Figure~\ref{fig:graphgameBS} for an illustration):

\begin{itemize}
  \item The set $N_\BS$ of the nodes (i.e., players) includes: a \emph{clause player} $c_j$, for each clause $c_j$; a \emph{literal player}
      $\ell_{i,j}$, for each literal $\ell_{i}$ occurring in $c_j$; and, two special players ``$chall$'' and ``$sat$''.

  \item The set $E_\BS$ of edges includes three kinds of edges.
\begin{description}
\item[\textbf{(Positive edges):}] an edge $\{c_j,\ell_{i,j}\}$ with $w(\{c_j,\ell_{i,j}\})=1$, for each literal $\ell_i$ occurring in the clause
    $c_j$; an edge $\{chall,\ell_{i,j}\}$ with $w(\{chall,\ell_{i,j}\})=1$, for each literal $\ell_i$ of the form $\alpha_i$ or $\neg
    \alpha_i$ (i.e., built over a universally quantified variable) occurring in $c_j$.

  \item[\textbf{(``Penalty'' edges):}] an edge $\{\gamma_{i,j},\neg \gamma_{i,j'}\}$ with $w(\{\gamma_{i,j},\neg \gamma_{i,j'}\})=-m-1$, for each variable $\gamma_i$
      (either $\gamma_i=\alpha_i$ or $\gamma_i=\beta_i$) occurring in $c_j$ and $c_{j'}$; an edge $\{\ell_{i,j},\ell_{i',j}\}$ with
      $w(\ell_{i,j},\ell_{i',j})=-m-1$, for each pair of literals $\ell_i$ and $\ell_{i'}$ occurring in $c_j$; an edge
      $\{chall,\ell_{i,j}\}$ with $w(\{chall,\ell_{i,j}\})=-m-1$, for each literal $\ell_i$ of the form $\beta_i$ or $\neg \beta_i$
      (i.e., built over an existentially quantified variable) occurring in $c_j$; an edge $\{chall,c_j\}$ with $w(\{chall,c_j\})=-m-1$, for each clause $c_j$.

  \item[\textbf{(``Normalizer'' edge):}] the edge $\{chall,sat\}$ with weight $w(\{chall,sat\})=n-1+m-\sum_{e\in E_\BS\mid e\neq\{chall,sat\}}w(e)$.
\end{description}
\end{itemize}

\begin{figure}[t]
\centering
\includegraphics[width=0.7\textwidth]{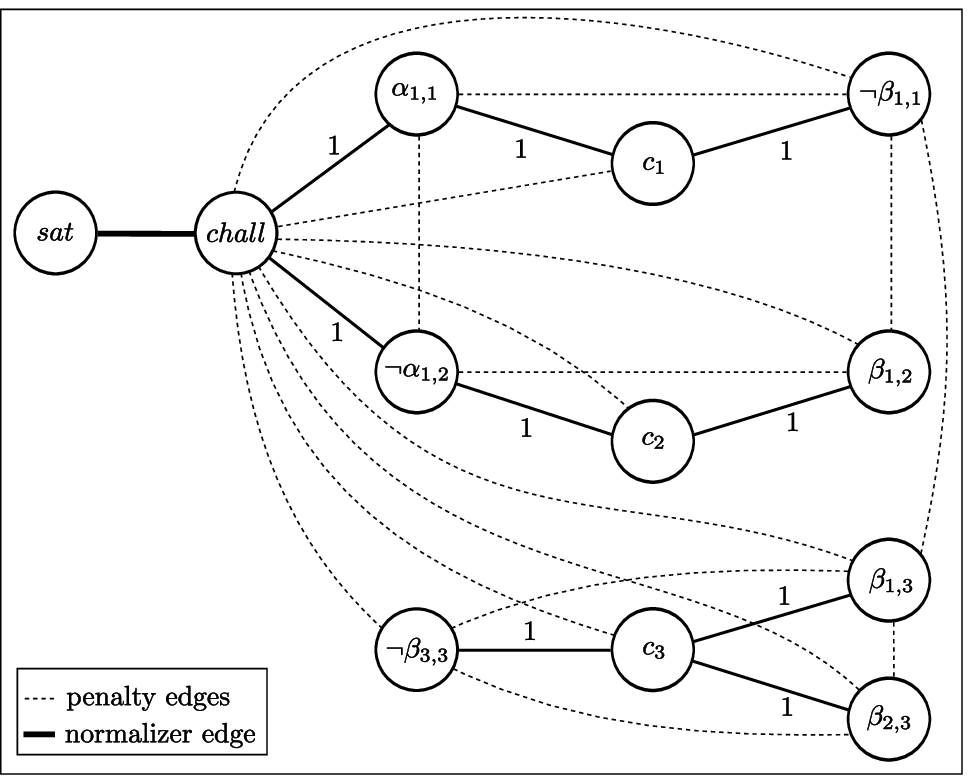}
\caption{The game $\qtugameBargSet(\widehat{\bs{\Phi}})$, where ${\widehat{\bs{\Phi}}}=(\forall \alpha_1)(\exists
\beta_1,\beta_2,\beta_3)(\alpha_1\lor \lnot \beta_1) \land (\lnot \alpha_1\lor \beta_1) \land
(\beta_1\lor\beta_2\lor\lnot\beta_3)$.}\label{fig:graphgameBS}
\end{figure}

Note that the size of the representation of all the weights is polynomial in the number of variables and clauses of $\bs{\Phi}$. Moreover, the
following properties hold.

\begin{lemma}\label{lem:bs}
Let $\qtugameBargSet(\bs{\Phi})=\tuple{(N_\BS,E_\BS),w}$ be the graph game associated with the
$\mathbf{NQBF}_{2,\forall}$ formula $\bs{\Phi}$. Then:
\begin{itemize}

\item[{\em (A)}] $D\leq m$;

\item[{\em (B)}] $w(\{chall,sat\})> 2\times m$;

\item[\em (C)] $D+w(e)<0$, for each penalty edge $e\in E_\BS$; and,

\item[\em {(D)}] $m\geq 2\times n$,
\end{itemize}
where $D=\max_{\{chall,sat\}\not\subseteq S\subseteq N} v(S)$ denotes the maximum worth over all the coalitions not covering the edge
$\{chall,sat\}$.
\end{lemma}
\begin{proof}
The fact (A) that $D\leq m$ is immediate by construction. Moreover, the weight of each penalty edge is $-m-1$ and, hence, (C) holds, too.
Eventually, (D) (i.e., $m\geq 2\times n$) is also immediate since $\bs{\Phi}$ is an $\mathbf{NQBF}_{2,\forall}$ formula.

Let us, hence, focus on (B) by observing that $\sum_{e\in E_\BS\mid e\neq\{chall,sat\}}w(e)\leq -m-1$, for $E_\BS$ contains more penalty edges
than positive ones.
Thus, $w(\{chall,sat\})=n-1+m-\sum_{e\in E\mid e\neq\{chall,sat\}}w(e)\geq n-1+m+(m+1)>2\times m$.
\end{proof}

\begin{theo}\label{theo:hardnessBargSetCheckingTU_Graph}
On the class $\C(\graph)$ of~graph games, \BargainingSetMembership\ is \PiPh2 (even if the given payoff vector
is an imputation).
\end{theo}
\begin{proof}
Deciding the validity of $\mathbf{NQBF}_{2,\forall}$ formulae is \PiPc2~\citep{Schaefer:GraphRamseyTheory}. Thus,
given an $\mathbf{NQBF}_{2,\forall}$ formula
$\bs{\Phi}=(\forall\bs{\alpha})(\exists\bs{\beta})\phi(\bs{\alpha},\bs{\beta})$ (where ${\bs
\alpha}=\{\alpha_1,\dots,\alpha_n\}$, ${\bs \beta}=\{\beta_1,\dots,\beta_r\}$, and
$\phi(\bs{\alpha},\bs{\beta})=c_1\wedge\dots\wedge c_m$), consider the graph game
$\qtugameBargSet(\bs{\Phi})=\tuple{(N_\BS,E_\BS),w}$ and the imputation $x$ that assigns $m$ to $sat$, $n-1$ to
$chall$, and $0$ to all other players.
Note that $x$ is an imputation, since $v(N)=m+n-1$ because of the weight of the edge $\{chall,sat\}$ (recall that
$w(\{chall,sat\})=m+n-1-\sum_{e\in E_\BS\mid e\neq\{chall,sat\}}w(e)$).

Beforehand, we note that the following properties hold on $\qtugameBargSet(\bs{\Phi})$ and $x$.

\begin{description}
\item[\textbf{Property \ref{theo:hardnessBargSetCheckingTU_Graph}.(1).}] \emph{No player has a justified objection against a clause or a
    literal player.}\\
Indeed, any clause player $c_j$ receives $0$ in $x$ and is such that $v(\{c_j\})=0$. Therefore, she can counterobject to any objection
through the singleton coalition $\{c_j\}$. Similarly, any literal player $\ell_{i,j}$ receives $0$ in $x$ and is such that
$v(\{\ell_{i,j}\})=0$, and can counterobject through $\{\ell_{i,j}\}$.

\item[\textbf{Property \ref{theo:hardnessBargSetCheckingTU_Graph}.(2).}]  \emph{No player has a justified objection against $chall$.}\\
Assume that a player $p\in N_\BS$ wants to object against $chall$ through a coalition $S$. If $v(S)<0$ then $p$ cannot object against
anyone. Indeed, she has to propose an $S$\nbdash feasible vector $y$ such that $y_k>x_k$ for all player $k\in S$, and hence
$v(S)=y(S)>x(S)$, but $x(S)\geq 0$ by definition, so player $p$ cannot fulfill this requirements if $v(S)<0$. Then, because of
Lemma~\ref{lem:bs}.(C), $S$ does not include penalty edges. In particular, for each universally quantified variable $\alpha_k$, it holds
that $|S\cap\{\alpha_{k,i(k)},\neg \alpha_{k,\bar i(k)}\}|\leq 1$, where $c_{i(k)}$ and $c_{\bar i(k)}$ are the two clauses where this
variable occurs (recall that $\bs{\Phi}$ is an $\mathbf{NQBF}_{2,\forall}$ formula).

Consider, then, the coalition $T\subseteq \{chall\}\cup \{ \alpha_{k,i(k)},\neg \alpha_{k,\bar i(k)} \mid 1\leq k\leq n\}$ such that
$|T|=n+1$, $T\cap S=\emptyset$, and $|T\cap \{\alpha_{k,i(k)},\neg \alpha_{k,\bar i(k)}\}|=1$, for each $1\leq k\leq n$. Note that $v(T)=n$
and $x(T)=x_{chall}=n-1$. Then, consider the vector $z$ such that $z_{chall}=x_{chall}=n-1$ and $z_q=\frac{1}{n}>x_q=0$ for each $q\in T$
with $q\neq chall$, and observe that $z(T)=v(T)$. Eventually, since $T\cap S=\emptyset$, $(z,T)$ is a counterobjection to any objection of
$p$ against $chall$ through $S$.

\item[\textbf{Property \ref{theo:hardnessBargSetCheckingTU_Graph}.(3).}]  \emph{No player different from $chall$
has a justified objection against $sat$.}\\
Suppose that a player $p\neq chall$ has an objection $(y,S)$ against $sat$ to $x$. Since $sat\not\in S$, it is the case that
$\{chall,sat\}\not\subseteq S$ and hence, by Lemma~\ref{lem:bs}.(A), that $v(S)\leq m$. Then, we claim that $(z,\{sat, chall\})$ is a
counterobjection to $(y,S)$ of $p$ against $sat$, where $z$ is a feasible distribution that assigns $m$ to $sat$ and $w(\{chall,sat\})-m$
to $chall$. In fact, by Lemma~\ref{lem:bs}.(B), $chall$ receives a payoff strictly greater than $m$ (i.e., $z_{chall}>m$). Then, note that $z_{sat}=x_{sat}$ and
let us distinguish two cases.

In the case where $chall\not\in S$, we have that $z_{chall}>m\geq n-1=x_{chall}$ (recall that $m\geq 2\times n$, by
Lemma~\ref{lem:bs}.(D)). Instead, in the case where $chall\in S$, we have that $z_{chall}>m$ while $y_{chall}\leq v(S)\leq m$, and hence
$z_{chall}>y_{chall}$. It follows that in both cases $(z,\{sat, chall\})$ is a counterobjection to $(y,S)$.
\end{description}

In the light of the properties above, we can limit our attention to the objections of $chall$ against $sat$. Consider an objection $(y,S)$ of
$chall$ against $sat$ to $x$. In particular, $y_{chall}$ must be greater than $x_{chall}=n-1$ and $y_q>0=x_q$ for each $q\in S$ with $q\neq
chall$. Thus, $y(S)=v(S)>n-1$.

Then, because of Lemma~\ref{lem:bs}.(C), in order for such a coalition $S$ to be such that $v(S)>n-1$, no penalty edge must be covered by $S$. Moreover, given that $chall\in S$, we have that $S$ must contain
exactly one player per universally quantified variable so that $v(S)=n>n-1$, i.e., $|S|=n+1$ and $|S\cap \{\alpha_{k,i(k)},\neg \alpha_{k,\bar i(k)}\}|=1$, for
each $1\leq k\leq n$. By this, the objection $y$ has to be a vector such that $y(S)=n$, $y_{chall}>n-1$, and $y_q>0$, for each $q\in S$ with $q\neq
chall$. Note that the coalition $S$ encodes an assignment $\sigma_S$ for the variables in $\bs \alpha$ such that $\alpha_k$ evaluates to false
(resp., true) in $\sigma_S$ if $\alpha_{k,i(k)}$ (resp., $\neg \alpha_{k,\bar i(k)}$) occurs in $S$---note that in this correspondence the
truth values are inverted with respect to the membership of the corresponding literal players in $S$.

Moreover, note that for each truth assignment $\sigma$ for the variables in $\bs \alpha$, we may immediately build a coalition $S$ and a vector
$y$ such that $(y,S)$ is an objection of $chall$ against $sat$, and $\sigma_S=\sigma$. Therefore, objections of $chall$ against $sat$ are in correspondence with truth assignments for universally quantified variables.

We are now ready to show that $\bs{\Phi}$ is valid if and only if $x$ is in $\bargset{{\qtugameBargSet(\bs{\Phi})}}$:
\begin{description}
\item[($\Rightarrow$)] Assume that $\bs{\Phi}=(\forall\bs{\alpha})(\exists\bs{\beta})\phi(\bs{\alpha},\bs{\beta})$ is valid, and let $(y,S)$
    be any objection of $chall$ against $sat$ to $x$. We show that this objection is not justified, because $sat$ has a counterobjection
    $(z,T)$. Recall first that $S$ encodes an assignment $\sigma_S$ over the variables in $\bs \alpha$. Then, let $\sigma$ be a satisfying
    assignment over the variables in $\bs{\Phi}$ such that $\alpha_k$ evaluates to true in $\sigma$ if and only if it evaluates to true in $\sigma_S$;
    indeed, such a satisfying assignment $\sigma$ exists since $\bs{\Phi}$ is valid.
    Based on $\sigma$, let us construct the coalition $T$ such that $T\subseteq\{sat\}\cup\{
    \ell_{i,j} \mid \ell_i\mbox{ evaluates to true in }\sigma\wedge {c_j}\mbox{ is a clause where $\ell_i$ occurs} \}\cup\{c_1,\dots,c_m\}$, $T\cap S=\emptyset$,
    $|T|=2m+1$, and $v(T)=m$. In particular, $T$ can be such that $v(T)=m$, precisely because $\sigma$ is a satisfying assignment and by
    construction of $\sigma_S$.
    Moreover, consider the vector $z$ such that $z_{sat}=v(T)=x_{sat}=m$ and $z_q=x_q=0$ for each $q\in T$
    with $q\neq sat$ (and, hence, $q\not\in S$). By construction $(z,T)$ is a counterobjection to $(y,S)$, which is therefore not
    justified in its turn.

\item[($\Leftarrow$)] Let $\sigma$ be an assignment over the variables in $\bs \alpha$ witnessing that $\bs{\Phi}$ is not valid. Let $(y,S)$ be an
    arbitrary objection of $chall$ against $sat$ to $x$ such that $\sigma_S=\sigma$. We claim that $(y,S)$ is justified. Indeed, assume for
    the sake of contradiction that $(y,S)$ is not justified and let $(z,T)$ be a counterobjection. Since $sat\in T\setminus S$ and since
    we must have $z_{sat}\geq x_{sat}=m$, because of Lemma~\ref{lem:bs}.(A) we actually have that $z_{sat}=x_{sat}$.
    In fact, since
    $m$ is the maximum available payoff over all the coalitions not including both $chall$ and $sat$, $z(T)=v(T)=m$ and the fact that $(z,T)$ is a counterobjection to $(y,S)$ together entail that $S\cap
    T=\emptyset$. Also, these entail that $T$ contains all clause players, exactly one literal player per clause, so that $|T|=2m+1$. Observe now that $T$ encodes a
    satisfying truth value assignment $\sigma_T$ for $\phi(\bs{\alpha},\bs{\beta})$,
    because $v(T)>0$ and hence $T$ does not cover any penalty edge.
    More precisely, we let each variable $\gamma_i$
    (either $\gamma_i=\alpha_i$ or $\gamma_i=\beta_i$) to evaluate true (resp., false) in $\sigma_T$ if $\gamma_{i,j}$ (resp., $\neg
    \gamma_{i,j}$) occurs in $T$, for some clause $c_j$. Let $\sigma_T^{\bs \alpha}$ denote the restriction of $\sigma_T$ over the
    variables in $\bs \alpha$. Then, since $S\cap T=\emptyset$ and given the definition of $\sigma_S$, we have that $\sigma_S=\sigma_T^{\bs
    \alpha}$. Clearly, this contradicts the fact that $\sigma=\sigma_S$ witnesses that $\bs{\Phi}$ is not valid.\qedhere
\end{description}
\end{proof}

\smallskip

Again, by the same argument used for the proof of Proposition~\ref{prop:hardness-core}, the above result extends to compact representations more expressive than graph games.

\begin{corol}\label{cor:hardness-bs}
Let $\R$ be any compact representation such that $\graph\precsim_e \R$ (e.g., $\R=\marginal$). On the class $\C(\R)$, \BargainingSetMembership\ is \PiPh2.
\end{corol}

\smallskip

\subsection{Membership on FNP Representation Schemes}

As for the result pertaining checking whether a payoff vector belongs to the bargaining set, it has already been
argued by \citeauthor{Deng_Papadimitriou:ComplexityCooperativeGameSolutionConcepts} that this problem is in \PiP2{} for graph
games~\citep{Deng_Papadimitriou:ComplexityCooperativeGameSolutionConcepts}. Indeed, they observed that
one may decide that an imputation $x$ is not in the bargaining set by firstly guessing in \NP{} the objection
$(y,S)$, and then by calling a $\CONP$ oracle, through which to check that there is no counterobjection $(z,T)$ to $(y,S)$.
However, to apply this argument one typically assumes to consider real values with
fixed\nbdash precision (or it should be proven that there always exist suitable objections and counterobjections that may be guessed in polynomial-time).

Our main achievement in this section is to show that membership in $\PiP{2}$ can be established, \emph{independently
of the precision adopted to represent the real values of interest in the game}.
 To this end, we shall provide a useful characterization of a player $i$ having a justified objection against some player $j$. The result is
in the spirit of one of~\citet{Maschler:InequalitiesBS}'s and connects the existence of a justified objection of $i$ against $j$ to
some algebraic conditions to hold on coalitions that $i$ and $j$ (may) belong to.

\begin{lemma}\label{lemma:condizioneOutsideBargSetTU}
Let $\game$ be a coalitional game, and let $x$ be an imputation of $\game$. Then, player $i$ has a justified
objection against player $j$ to $x$ through coalition $S\in \mathcal{I}_{i,j}$ if and only if there exists a
vector $y\in\Re^{S}$ such that:
\begin{enumerate}
\item[\em (1)] $y(S)=v(S)$;

\item[\em (2)] $y_k>x_k$, for each $k\in S$; and,

\item[\em (3)] $v(T)< y(T\cap S)+x(T\setminus S)$, for each $T\in\mathcal{I}_{j,i}$.
\end{enumerate}
\end{lemma}
\begin{proof}\
\begin{description}
\item[($\Rightarrow$)] Assume that player $i$ has a justified objection against player $j$ to $x$ through coalition $S\in
    \mathcal{I}_{i,j}$. Let $(y,S)$ be such justified objection and note that, by definition, $y(S)=v(S)$ and $y_k>x_k$ for each $k\in S$
    hold. Assume now, for the sake of contradiction, that {(3)} above does not hold. Let $\bar T\in\mathcal{I}_{j,i}$ be such that $v(\bar
    T)\geq y(\bar T\cap S)+x(\bar T\setminus S)$. Based on $\bar T$, let us build a vector $z\in \Re^{\bar T}$ such that:
    \emph{(i)} $z_k=x_k+\delta$ for each $k\in \bar T\setminus S$; and, \emph{(ii)} $z_k=y_k+\delta$ for each $k\in \bar T\cap S$, where
$$\delta=\frac{v(\bar T)- y(\bar T\cap S)-x(\bar T\setminus S)}{|\bar T|}\geq 0.$$ Note that $z(\bar T)=v(\bar T)$
holds by construction. Hence, $(z,\bar T)$ is a counterobjection to $(y,S)$, a contradiction.

\item[($\Leftarrow$)] Assume that there is a vector $y\in \Re^{S}$ such that all the three conditions above hold. Consider the pair
    $(y,S)$ (with $S\in \mathcal{I}_{i,j}$) and note that due to ({1}) and ({2}), $(y,S)$ is in fact an objection of player $i$ against
    player $j$. We now claim that $(y,S)$ is justified. Indeed assume, for the sake of contradiction, that a counterobjection $(z,\bar T)$
    (with $\bar T\in\mathcal{I}_{j,i}$) exists such that: $z(\bar T)=v(\bar T)$, $z_k\geq x_k$ for each $k\in \bar T\setminus S$, and
    $z_k\geq y_k$ for each $k\in \bar T\cap S$. In this case, it holds that $v(\bar T)=z(\bar T)\geq y(T\cap S)+x(T\setminus S)$, which
    contradicts (3).\qedhere
\end{description}
\end{proof}

\begin{theo}\label{theo:membershipBargainingSetCheckingTU}
Let $\mathcal{R}$ be a non\nbdash deterministic polynomial\nbdash time compact representation. On the class
$\C(\mathcal{R})$, \BargainingSetMembership\ is feasible in \PiP2.
\end{theo}
\begin{proof}
Let $\game\in\C(\mathcal{R})$ be a coalitional game, and $x$ be a payoff vector. W.l.o.g., we assume that $x$ is an imputation, since checking this condition is feasible in \PiP2 (see
the proof of Theorem~\ref{theo:membershipKernelCheckingTU}).

Consider the complementary problem of deciding whether $x\not\in \bargset{\game}$. We are going to show that
this is feasible in \SigmaP2. In the light of Lemma~\ref{lemma:condizioneOutsideBargSetTU}, $x\not\in\bargset{\game}$ if and only if there
exist two players $i$ and $j$, and a coalition $S\in \mathcal{I}_{i,j}$ such that the set
\begin{align*}
W(i,j,S)=\{y\in \Re^{S} \mid &y(S)=v^\mathcal{R}(\xi^\mathcal{R}(\game),S) \wedge\\
                             &y_k>x_k,\ \forall k\in S \wedge\\
                             &v^\mathcal{R}(\xi^\mathcal{R}(\game),T)< y(T\cap S)+x(T\setminus S),\ \forall T\in\mathcal{I}_{j,i}\}
\end{align*}
is not empty.
Thus, the problem can be solved by guessing in $\NP$ the players $i$ and $j$ and the set
$S\in\mathcal{I}_{i,j}$, and then by asking an \NP{} oracle to decide whether $W(i,j,S)$ is empty or not. In the latter case, we conclude that $x\not\in\bargset{\game}$.

In order to show that this set is empty, the \NP{} oracle may proceed as follows:
it guesses at most $m\leq |S|+1$ coalitions $T_1,\ldots T_m\in \mathcal{I}_{j,i}$,
 plus their worth values $w_1,\ldots,w_m$, and the associated certificates to check that actually $\tuple{(\xi^\mathcal{R}(\game),T_\ell),w_\ell}\in graph( v^\mathcal{R})$, $\forall 1\leq\ell\leq m$. Moreover, it guesses the worth value $v^\mathcal{R}(\xi^\mathcal{R}(\game),S)$ and a suitable certificate $c_S$ for it. Then, by exploiting these certificates, the oracle checks in (deterministic) polynomial time that all guessed worth values are correct, as well as that the coalitions $T_1,\ldots , T_m$ belongs to $\mathcal{I}_{j,i}$ and that $S$ belongs to $\mathcal{I}_{i,j}$.
   Moreover, it computes the inequalities defined by the first two rows in the above expression, plus the inequalities associated with the guessed coalitions
    $T_1,\ldots ,T_m\in \mathcal{I}_{j,i}$.
   Finally, it checks in (deterministic) polynomial time whether this sub-system of at most
    $2(|S|+1)\leq 2n$ inequalities has some solution or not. In the latter case the oracle accepts, that is, it answers that $W(i,j,S)=\emptyset$.

To see that this procedure is correct, note that $W(i,j,S)$ is defined by the intersection of a finite number of convex regions of $\Re^{S}$ and hence, by Proposition~\ref{theo:hellyTheo} (Helly's theorem), if $W(i,j,S)$ is empty then there exist a set $\mathcal{E}$ including (at most) $|S|+1$ of its defining convex regions (here, inequalities) witnessing its emptiness. Therefore, if this is the case, there exists some accepting computation of the oracle where the inequalities associated with the $m\leq |S|+1$ guessed coalitions plus the $|S|+1$ inequalities associated with the first two rows of the above expression includes $\mathcal{E}$, and thus the considered sub-system of inequalities has no solutions and it is a (succinct) witness that $W(i,j,S)=\emptyset$.
\end{proof}

From this result and Corollary~\ref{cor:hardness-bs}, we immediately get the following completeness result.

\begin{corol}
Let $\R$ be any non-deterministic polynomial\nbdash time compact representation such that $\graph\precsim_e \R$ (e.g., $\R=\marginal$). On the class $\C(\R)$, \BargainingSetMembership\ is \PiPc2.
\end{corol}

\section{Conclusions}\label{sec:conclusion}

In this paper, we have provided a complete picture of the complexity issues arising with the notions of core,
kernel, and bargaining set on compactly specified coalitional games. On the one hand, we have exhibited membership results
that apply to any arbitrary \NP{} representation scheme for coalitional games. On the other hand, the
corresponding hardness results have been proven on the setting of graph games (which is a very simple \PP{}
representation scheme); hence, these hardness results apply to marginal contribution networks and to any
representation scheme that is at least as expressive as graph games. Our results confirmed two conjectures by \citet{Deng_Papadimitriou:ComplexityCooperativeGameSolutionConcepts}, and positively answer an
open question posed by \citet{Ieong_Shoham:MarginalContributionNets}.

Since the three analyzed solutions concepts turned out to be computationally intractable in general, it is worthwhile looking for natural and expressive enough classes of games where these solutions concepts  can be efficiently computed. This question has received only partial answers so far. It has been shown, by \citet{Ieong_Shoham:MarginalContributionNets}, that deciding core non-emptiness and deciding
whether a payoff vector belongs to the core are problems feasible in polynomial time on classes of marginal
contribution networks having bounded treewidth. In this paper, we have further explored this avenue of research,
by showing the tractability of the kernel over graph games having bounded treewidth. However, the tractability
frontier for such solution concepts is far from being completely charted. In particular, it remains an open problem to establish whether the
kernel remains tractable on classes of marginal contribution networks having bounded treewidth, and whether the bounded
treewidth is also the key for establishing tractability results when the bargaining set is considered.

\bibliographystyle{abbrvnat}
\bibliography{ref}
\end{document}